\documentclass[12pt, draftcls, onecolumn]{IEEEtran}  % Comment this line out
\IEEEoverridecommandlockouts                              % This command is only
\usepackage{textcomp}
\usepackage{ifthen}
\usepackage{amsfonts}
\usepackage{amsmath}
\usepackage{amssymb}

 \usepackage{amsthm}
\usepackage{verbatim}
\usepackage{cite}
\usepackage{subfigure}
\usepackage{cases}
\usepackage[usenames]{color}
  \usepackage{flushend}
\usepackage{graphicx} % support the \includegraphics command and options

\usepackage{booktabs} % for much better looking tables
\usepackage{array} % for better arrays (eg matrices) in maths
\usepackage{paralist} % very flexible & customisable lists (eg. enumerate/itemize, etc.)
\usepackage{verbatim} % adds environment for commenting out blocks of text & for better verbatim
\usepackage{subfig} % make it possible to include more than one captioned figure/table in a single float
\usepackage{cases}

\newcommand \remove[1] {}
\newcommand \hide[1] {}
\newcommand\red[1]{{\color{red} #1}}
\def\Sec {\S}

\def \ham {\mathcal{H}}

\def \hence {\Rightarrow}

\newtheorem{Theorem}{Theorem}

\usepackage{tikz}
\usetikzlibrary{
	arrows,decorations.pathmorphing,backgrounds,positioning,fit,calc,scopes, snakes, backgrounds, automata
}
 
\tikzset{
	auto,
	compartment/.style={
		rectangle, minimum size=6mm, rounded corners=2mm,
		thick, draw=black!15, top color=white,bottom color=black!30
	},
	bigcompartment/.style={
		rectangle, minimum size=9mm, rounded corners=2mm,
		thick, draw=black!15, top color=white,bottom color=black!20
	},
	point/.style={
		circle, inner sep=2pt, fill=black!5
	},
	mytextbox/.style={
		rectangle, text=black!50, thin, 
		draw=white, top color=white,bottom color=white, fill=white
	}

}

\newcommand{\sirboxes}{
	\node(S)[compartment]{S};	
	\node(SI)[right of= S]{};
	
	\node(I)[compartment, right=of SI]{G};

	\node(Z)[compartment, above of= I]{Z};
	\node(P)[compartment, below of= I]{P};
      \node(SII)[right of= I]{};
}
\newcommand{\sirvec}{

\path[->, very thick] (S) edge [ align=left] node {\small Zombie creation\\ \small
$\gamma \beta  Z S+ \beta G  S u_Z$} (Z);

\path[->, very thick] (S) edge [swap,align=right] node {\small$\beta G S u_p$ \\\small Passive creation } (P);

\path[->, dotted, purple, very thick,bend left=90] (Z) edge [align=left] node {\small Halting zombies\\\small$\pi \beta G Z u_h$} (P);
}

\newcommand{\ispoint}{

\path[->, green, very thick, bend right=30] (I) edge ($(P.west)!0.5!(S.east)$);
\path[->, green, very thick, bend left=30] (I) edge ($(Z.west)!0.5!(S.east)$);
\path[->, green, very thick, bend right=60] (Z) edge ($(Z.west)!0.5!(S.east)$);
\path[->, green, dotted, very thick] (I) edge (SII);

}

\title{\LARGE \bf
Visibility-Aware Optimal Contagion of Malware Epidemics
}

\author{Soheil Eshghi, Saswati Sarkar, and Santosh S. Venkatesh\thanks{Soheil Eshghi
is with the Yale Institute of Network Science (YINS), Yale University, New Haven, CT, {\tt\small soheil.eshghi@yale.edu}, Saswati Sarkar, and 
Santosh S. Venkatesh are with the Electrical and Systems Engineering Department, University of Pennsylvania, Philadelphia, PA, {\tt\small swati, venkates@seas.upenn.edu}.
}
}

\begin{document}

\maketitle
\thispagestyle{plain}
\pagestyle{plain}

%%%%%%%%%%%%%%%%%%%%%%%%%%%%%%%%%%%%%%%%%%%%%%%%%%%%%%%%%%%%%%%%%%%%%%%%%%%%%%%%
\begin{abstract}
Recent innovations in the design of computer viruses have led to new trade-offs for the attacker. Multiple variants of a malware may spread at different rates and have different levels of visibility to the network. In this work we examine the optimal strategies for the attacker so as to trade off the extent of spread of the malware against the need for stealth. We show that in the mean-field deterministic regime, this spread-stealth trade-off is optimized by computationally simple single-threshold policies. Specifically, we show that only one variant of the malware is spread by the attacker at each time, as there exists a time up to which the attacker prioritizes maximizing the spread of the malware, and after which she prioritizes stealth. 
\end{abstract}
%%\vspace{-0.1in}
\vspace{-0.2in}
%%%%%%%%%%%%%%%%%%%%%%%%%%%%%%%%%%%%%%%%%%%%%%%%%%%%%%%%%%%%%%%%%%%%%%%%%%%%%%%%
\section{Introduction}\label{sec:intro}

Malware (i.e., viruses, worms, trojans, etc.) has been a prominent feature of computer networks since the 1980's \cite{murray1988application}, and has evolved with the growing capabilities of computing technology. Anderson {\em et al.} \cite{anderson2012measuring} estimated that malware caused \$370m of damage globally in 2010 alone. Traditionally, malware was designed with the express aim of infecting as many machines as possible, leading to the mass epidemics of the early 2000's (e.g., Blaster\cite{bailey2005blaster}). More recently, the focus has shifted to more ``surgical'' strikes where visibility is highly undesirable, as awareness can lead the intended target to cease communication (e.g., by quarantining the targets). The malware Regin was only discovered (in 2014) after operating since at least 2008, and was so complex that even when its presence was detected, it was not possible to ascertain what it was doing and what it was targeting\cite{Regin}. Stuxnet, as another example, was designed to attack a specific control software used in centrifuges \cite{falliere2011w32} and did not steal or manipulate data, or receive any command instructions from remote sources so as to maintain stealth \cite{langner2013stuxnet}. {Furthermore, its very presence in a system was undetectable due to a rootkit\cite{falliere2011w32}.} Yet, it was discovered and remedied after it spread outside its target area \cite{nyt} (cf. Duqu, Flame, and Gauss \cite{bencsath2012cousins}). Thus there is a new {critical} trade-off for the attacker --- to ensure maximum damage while minimizing visibility to the defender.  

 We now describe different dimensions of this trade-off. Malware spreads from one computing device to another when there is a communication opportunity between the devices. In networks, both wired and wireless, inter-node communication can be visible to the network administrator, and can serve as a way of detecting the presence of malware before its function is fully understood.
However, the attacker also has a conflicting onus to ensure the rapid propagation of her program, as computer systems evolve at a rapid pace, and the exploit(s) that the malware targets will be noticed and patched in due course. Furthermore, some malware designers work to specific deadlines  --- e.g., Stuxnet was due to become inoperational in June 2012 \cite{zetter2011stuxnet}. {On the other hand, the second variant of Stuxnet was released to spread faster (and thus in a more risky manner) after the designers were concerned about its limited spread \cite{langner2013stuxnet}.}  Thus, an attacker will seek to minimize her communication footprint while still trying to ensure the timely spread of the malware. 

{In particular, we consider the case where two variants of a single emerging malware spread in a network that is unaware of their existence}. One spreads aggressively in every contact, and is thus visible to the network due to its communications, while the other, passive, variant does not spread subsequent to infecting a node. {We assume that the network cannot determine the infection state of any particular node and does not have patches to remedy the attack, but can detect an attack by looking at the unusual communication patterns (e.g., the transfer of malware between nodes) resulting from the malware attack.} Coordinating distributed attacks comes at the cost of added visibility due to communication and is susceptible to timing errors in the hosts. Thus, we focus on the case where distributed nodes that are infected are not asked to coordinate, as was the case in Regin and Stuxnet.
The natural question that arises is to  characterize the structure of the optimal malware variant mix that the attacker will spread at each instant depending on their goal structures and the communication mechanisms that they may have at their disposal. This is an imperative first step to devising remedies for such attacks.

\vspace{-0.1in}
\subsection{Problem Description}
 We  consider a network under attack by these two variants of a malware. Depending on their infection status, nodes can be divided into 4 groups:\footnote{{Note that this classification and the resulting dynamics are an abstraction of real world networks and sacrifice some accuracy for modeling simplicity. However, these assumptions are common in cybersecurity literature, e.g., \cite{piqueira2005epidemiological, murray1988application} and lead to significant insight.}} Germinators (G), Susceptibles (S), Zombies (Z), and Passives (P). We now describe these states, as well as their dynamics and the impact of the attacker's control (as will be elucidated in \S \ref{subsec:simple_model}). We also outline an augmentation to the model that is considered in \S \ref{subsec:killing_model} and adds a further possible mechanism of interaction and control to the dynamics:

1) {\em Germinators} (G):

- are a {\em fixed} (potentially very small) fraction of nodes,
- are the only nodes under the attacker's direct control,

- are the only nodes that can {\em choose} how to interact with susceptibles and zombies depending on the goal of the attacker: at each encounter with a susceptible, they decide whether to turn it into a zombie or a passive, or to leave it as a susceptible.

- damage the network by executing malicious code,

- are visible to the network due to their communications. 

- in an augmentation in \S \ref{subsec:killing_model}, we add a further mechanism of interaction ({\em halting}) whereby the germinators, upon contact with zombies, can turn them into passives (i.e., stopping them from spreading the message any further). This can potentially lead to the attacker initially utilizing epidemic spreading and then halting the spread once the marginal benefit of infection is overtaken by the marginal effect of visibility, leading to to a potentially longer propagation of the zombies. 

2) {\em Susceptibles} (S):

- are nodes that have not received any variant of the malware, 

- upon receipt of the malware from germinators, they can turn into zombies ($Z$) or passives ($P$).

- upon receipt of the malware from zombies, they will turn into zombies ($Z$).

3) {\em Zombies} (Z):

- have received the aggressive malware variant, 

- damage the network by executing malicious code,

- will continue to propagate the aggressive variant indiscriminately (i.e., upon meeting a susceptible, will turn into a zombie), 

- are visible to the network due to their communications.

- in the augmentation in \S \ref{subsec:killing_model}, the additional mechanism of {\em halting} can turn zombies into passives.

4) {\em Passives} (P):

- have received the passive variant of the malware,

- damage the network by executing malicious code,

- will {\em not} propagate the malware variant any further, 

- contrary to germinators and zombies, are {\em invisible} to the network as they do not communicate with other nodes to spread the malware henceforth.

These  states and their properties are summarized in Table \ref{fig:status}. {We assume that all nodes mix \emph{homogeneously} (i.e., contacts between nodes are independent and exponentially distributed) with rates that only depends on the infection states of the two nodes. Thus, all nodes that are in one infection state can be assumed to be identical from the perspective of the malware. The purpose of this abstraction is to simplify the interaction model for analysis in the population limit (i.e., as the number of nodes increases).}

\begin{table}
\centering
\begin{tabular}{|c||c|c|c|}
\hline 
 State & Visibility & Growth over time & Propagation
\tabularnewline 
\hline
\hline
 S & N & Only decrease & -
\tabularnewline 
\hline
 G & Y & Fixed & Y
\tabularnewline 
\hline
 Z & Y & Increase or decrease & Y
\tabularnewline 
\hline
 P & N & Only increase & N
\tabularnewline 
\hline
\end{tabular}
\caption{The states of the SGZP model and their characteristics. ``Visibility'' denotes whether the infection state of the node is detectable by the network defender. ``Growth over time'' determines the possible changes in the fraction of nodes in each state over time (note that the only case in which zombies can decrease is the dynamics outlined in \S \ref{sec:Optimal_killing}). Finally, ``Propagation'' determines whether a node in that state can spread the malware to a susceptible node upon contact.}\label{fig:status}
\end{table}

In these models, the attacker {\em controls} the mixture of zombie and passive malware variants through the germinators under its direct control. Whenever a germinator meets a susceptible, based on the control chosen by the attacker, it spreads either the zombie or passive variant of the malware to the susceptible, or leaves it as it is. In the dynamics in \S \ref{sec:Optimal_killing}, the germinator has an additional controlled mechanism of action, whereby upon meeting a node with the zombie variant of the malware, it can replace the variant with the passive one (a ``halting'' mechanism). These controls are assumed to be {\em piecewise continuous}, but they can take any value between zero and one, which determines the percentage of relevant interactions for which the specified action happens.
We do not assume that all nodes make the same spreading decision at each time instance: the attacker can assign a certain uniformly distributed and possibly varying fraction of germinators to make the same decision at each time, or it could allow all agents to make one of the two decisions with a certain, possibly varying probability at each time. The outcome of both cases is that a certain uniformly distributed percentage of interactions (derived from the attacker's controls) lead to the creation of zombies and passives, and the rest have no effect on the potential target.

 Later, we also investigate the effect of defense strategies on the optimal spread of malware variants (\S \ref{subsec:beta_model}). In these defense strategies, the defender  limits the effective contacts of nodes using a pre-determined function of malware visibility (which changes over time) as a means to limit the spread of malware. We consider two classes of network defense functions: affine and sigmoid. These defense strategies, however, come at the cost of stopping legitimate communication within the network. This is akin to choosing the communication ranges of nodes as a decreasing function of the visibility of the malware, which is a form of {\em quarantine}. 

We allow the attacker to choose the malware spreading controls so as to maximize a measure of overall damage (described in \S \ref{subsec:cost}). 
We first consider a damage function that depends on a) malware efficacy, which is a function of the aggregate number of zombies and passives, and b) malware visibility, which is a function of the number of zombies (for the models in \S \ref{subsec:simple_model} and \S \ref{subsec:killing_model}). Then, we consider a damage function where malware efficacy is the attacker's only direct concern, and is thus the damage function to be maximized, for the case where visibility is built into the network dynamics through a network defense policy which is a function of the fraction of zombies (as in the model in \S \ref{subsec:beta_model}). These formulations, to the best of our knowledge, have no precedent in the epidemics literature, and can be used to further investigate the effects of  malware visibility in networks. 

 An advantageous feature of all these models is that the malware designer only requires synchronized actions from a fixed number of nodes that are under its control from the outset. This decreases the risks of detection and policy implementation errors arising from coordinating synchronized distributed actions among a varying set of nodes. 
\vspace{-0.1in}
\subsection{Results}
We then derive necessary structures for optimal solutions for each of the cases, using Pontryagin's Maximum Principle and custom arguments constructed for each case (in \S \ref{sec:structural}). We show that the attacker's optimal strategy in all of these models is for the germinators to spread only one variant of the epidemic at each time: the germinators will create zombies up to a certain threshold time, and then only create passives (including by halting zombies) from then on. That is, the optimal controls are {\em bang-bang} (i.e., only taking their minimal and maximum values) with only one jump. Note that the controls can take any value between 0 and 1 at each point in time, and this bang-bang structure is one that emerges from the dynamics of the problem. These structural results are without precedent in the literature, both due to the uniqueness of the model, as well as the constraints placed on the {\em vector} of optimal controls.

It is interesting to note that in each of the variations we consider, our analysis reveals that all the controls in each model have the same threshold, a fact that is not at all clear {\em a priori}. Thus the entire control space can be described by one time threshold. This structure is invaluable for deriving the optimal controls computationally {(by solving the scalar optimization problem with the state ODEs mapping the variable to the damage objective)}. Furthermore, the controls are {deterministic and} easy to implement as the germinators need to be programmed with just one time instant for all of their controls. 

Finally, we investigate the performance of the derived optimal controls using numerical simulations (in \S \ref{sec:heuristics}). We first investigate the effect of the additional halting action on the optimal attack policies. We show that for both the simple and halting models, as the rate of contact between zombies and susceptibles increases, zombies are created for a shorter time period. We also show that the halting control adds to the length of time the zombie variant should optimally be propagated, with the additional propagation time depending on some system parameters. We then compare the optimal control with heuristics, and show that even without the halting control, the optimal solution performs $10\%$ better than the leading heuristic, with the performance differential being larger for more naive heuristics. We then consider errors in the implementation of the network defense strategy outlined in \S \ref{subsec:beta_model}, and investigate their effects on the malware spread. We show that erroneous estimations on the part of the defender only slightly  affect the damage inflicted by the attacker, which points towards the robustness of the attack policies to errors in estimations by the network defense. Finally, we quantify the effect of synchronization errors among the relatively small number of germinators on the efficacy of the malware attack. We show that any such attack is robust to small errors among the germinators, sounding an alarm to the fact that these malware attacks are less vulnerable to implementation issues that may arise from synchronization errors than previous generations of malware.
\vspace{-0.1in}
\section{Literature Review}\label{sec:Lit_Review}

Multiple interacting epidemics that spread among a single population have been considered in the fields of biology (e.g., multiple strains of a viral epidemic \cite{karrer2011competing, kendall1983epidemics}) and sociology (e.g., competition among memes in a world with limited attention span\cite{weng2012competition}). 
The {\em key} distinction between the control of biological epidemics \cite{hethcote1973optimal,morton1974optimal, wickwire1975optimal, wickwire1976optimal, behncke2000optimal} and that of malware ones is that in malware epidemics the {\em attacker} can also decide to {\em use her resources optimally} and to {\em adapt} to foresee the response of the defender. In the realm of sociology, the control of information epidemics offers closer parallels to that of malware. For example, Kandhway and Kuri \cite{kandhway2014run} model how an erroneous rumor may be optimally stifled by the spread of correct information, which is a secondary epidemic that interacts with the naturally occurring rumor epidemic. However, in this case {\em only one} of the epidemics can be controlled, while the malware attacker can possibly simultaneously control the spread of {\em all} malware variants. When there are multiple controllable epidemics, the resulting simultaneous controls are interdependent, and focusing on one control and characterizing its structure does not lead to a characterization of the optimal action. Thus, in malware epidemics there are {\em vectors} of controls available to the attacker, which requires new approaches and techniques compared to the other fields discussed.

Even within the majority of malware epidemic models, e.g., \cite{kim2004measurement, zou2002code,garetto2003modeling, lu2014botnets, hall2006thwarting,schramm2013lanchester, khouzani2011maximum}, the spread of {\em only one} malware has been examined, while we focus on the case where two variants are spreading in conjunction with each other. This presents a fundamentally different choice to the attacker, and so the model presented for the spread of visibility-heterogeneous malware variants has no precedent in literature. Accordingly, the questions we asked and the solutions we obtained are substantially different to prior work. 

{Note also that in nearly all malware epidemics, as well as the more generic epidemic models mentioned, some form of the homogeneous mixing assumption is used to obtain tractable results. While \cite{watkins2015deterministic} provides one interesting avenue for the relaxation of the mean-field assumption in the study of a given epidemic process, tractable results in the epidemic control domain still critically rely on the mean-field assumption.}

Nonetheless, we still distinguish other aspects of our work from those considering a single type of malware: in these papers: 1- it is assumed that the attacker's sole aim is to maximize the spread of the malware, which is no longer the case for the emerging class of surgical malware such as Regin\cite{Regin} and Stuxnet\cite{falliere2011w32} and {2- attackers have a mechanism to control the spread of the malware remotely in the future, e.g., through a timer in the code which would be executed in infected machines (as in \cite{eshghi2014optimal2}). Any such code would have to interact with the operating system of the infected node, the configuration of which might not be known to the attacker, and can thus create a point of failure for the malware.} The failure of such a mechanism of control was key to the overspread and subsequent remedy of Stuxnet \cite{nyt}.

 Among the work on the control of a single-type/variant of malware (and the closely related literature on the spread of a message in Delay Tolerant Networks \cite{eshghi2015optimal, singh2010delay} and the spread of a rumor\cite{kandhway2014run}), the closest work to this topic (in terms of approach and spreading models) was {in two papers \cite{khouzani2011maximum, khouzani2010dynamic}. In both papers, however, the authors
assume that the malware can control the transmission range of infected nodes\footnote{We assume that the control affects the mix of malware variants and that the communication ranges of nodes are outside the malware's control, perhaps even being controlled by the defender as a mitigation mechanism. Thus, the control and the trade-off to the malware designer is fundamentally different.} and {\em patching} is the major defense of the network and starts as soon as the epidemic spreads\footnote{This may not be the case for an emerging stealthy epidemic like Stuxnet that is very large and extremely hard to decipher, let alone mitigate\cite{zetter2011stuxnet, chen2011lessons}. In our model, the network only becomes aware of the malware as it becomes more visible (i.e., as the visible variant spreads).}. Thus, while the derived bang-bang structure of the optimal controls is similar, their models and their results apply to a fundamentally different class of malware, and the arguments used in deriving the results are only similar at the level of using a classic Maximum Principle-derived switching function argument for constrained controls. Furthermore, the adaptive defense model and the results on the simultaneity of 3 optimal control switching times for the halting model are without precedent in the literature.}

Finally, the very strict structure we prove for the vector of malware optimal control, which restricts the search space for computational methods to a single parameter, is also without precedent in any of the aforementioned literature.

\vspace{-0.1in}
\section{System Model and Objective Formulation}\label{sec:Sys_Model}

In this section we model the spread of malware in a homogeneous network with random contacts. This can be the case where malware spreads among mobile devices with proximity-based communication, or where random contacts in an address-book are utilized. The virus propagates in the network between times 0 and $T$. We represent the fraction of {\em susceptible}, {\em germinator}, {\em zombie}, and {\em passive} nodes at time $t$ with $S(t)$, $G(t)$, $Z(t)$, and $P(t)$ respectively, and assume that they are differentiable {functions of time}. We assume that for any pair of states, the statistics of meeting times between all pairs of nodes of those two states are identical and exponentially distributed, where the mean is equal to the {\em homogeneous mixing rate} of those two states.  Groenevelt {\em et al.}\cite{groenevelt2005message} have shown that homogeneous mixing holds under the common Random Way-point and Random Direction mobility models (when the communication range of the fast-moving nodes is small compared to the total region). {It has been shown \cite[p.1]{kurtz1970solutions, gast2010mean} that the resulting evolution of such a set of state fractions (where state transitions occur according to a Poisson contact process) will converge pathwise to the solution of a set of ordinary differential equations derived from the dynamics in the population limit (i.e., in the mean-field) on any limited time period (in particular, including the transient phase). In previous work, we have shown that such approximations are reasonable even with populations as small as 40-160 \cite{eshghi2015optimal}.\footnote{This work \cite{eshghi2015optimal} also lays out a roadmap on how to partially relax the homogeneous mixing assumption in the current problem.} }

Note that the zombies can be programmed to only spread the malware at a fraction of the times they meet susceptibles, slowing their spread, or they can be programmed to use resources that are not utilized by the rest of the network to spread faster. Therefore we take the mixing rate between Z and S to be potentially different from the other pairs of states. 

We describe the state dynamics of such systems as an epidemic for the cases where: 1) germinator agents can only interact with susceptible agents (\Sec\ref{subsec:simple_model}), 2) germinator agents can also interact with zombies as well (\Sec\ref{subsec:killing_model}), and 3) effective network contact rates are a function of the infection spread, mirroring the response of a network defender (\Sec\ref{subsec:beta_model}) (Figure \ref{fig:SIZP}). We state and prove a key observation about all these dynamics (\S \ref{subsec:observations}). We next formulate the aggregate damage of attack efficacy and the ensuing visibility (\Sec\ref{subsec:cost}). Finally, we lay out the optimization problem in \S \ref{subsec:statement}. 

\begin{figure}
\centering
 \begin{tikzpicture}[font=\LARGE]
 
\sirboxes
\sirvec
\ispoint

\end{tikzpicture}
 \caption{{The blocks represent the 4 states of nodes with regard to the malware. The solid black lines show the dynamics in \S \ref{subsec:simple_model} with the transition rates super-imposed. The green arrows point from each source of  malware to the resulting transition. The dotted red lines show the additional halting action in \S \ref{subsec:killing_model}. The model in \S \ref{subsec:beta_model} has the same dynamics as the solid black lines, but with $\beta$ being a function of $Z$ (i.e., $\beta(Z)$).}}
 \label{fig:SIZP}
\end{figure}
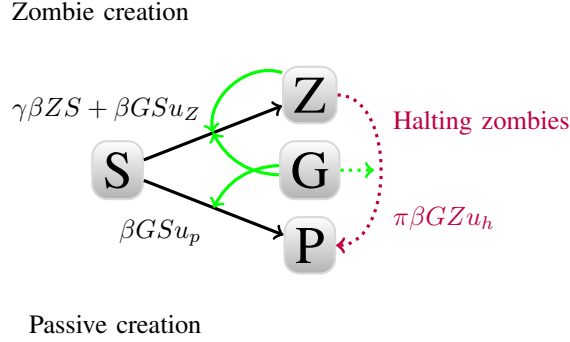

\vspace{-0.15in}
\subsection{SGZP Model with no halting}\label{subsec:simple_model}

 The attacker can spread the malware in two ways: 1- upon encountering a susceptible, she can, through the control variable $u_Z(t)$, turn that susceptible node into a zombie, {i.e., one that will henceforth propagate that infection to susceptibles it meets}. 2- upon encountering a susceptible, she can, through the control variable $u_P(t)$, turn that susceptible into a \emph{Passive}, $P$. These control variables  --- $(u_Z, u_P)\in \mathcal{U}$, where $\mathcal{U}$ is the set of piecewise continuous controls --- can be thought of as the {probabilities that an interaction of a germinator and a susceptible at time $t$ will lead to the susceptible becoming a zombie and a passive respectively. To maintain such a probabilistic intuition, we constrain their sum to be less than one.}
\vspace{-0.05in}
\begin{subequations}\label{model:nokilling}
\begin{align}
\dot{S} &=-\beta G S (u_P+ u_{Z}) -\gamma\beta Z S\\
\dot{Z} &= \beta G S u_{Z} + \gamma\beta Z S \label{eq:Z_nokilling} \\
\dot{P} &=\beta G S u_P
\end{align}
\end{subequations}
\vspace{-0.3in}
\begin{subequations}\label{nohalt_bounds}
\begin{align}
u_P+ u_{Z} \leq 1\label{sum_bound_nohalt}\\
0 \leq u_P \leq 1~~
0 \leq u_{Z}\leq 1\label{bound_nohalt}
\end{align}
\end{subequations}
Here, $\beta$ is the mixing rate between $S$ and $G$ {(which the attacker can calculate using time averages of contact times)}, and $\gamma\beta$ is the mixing rate between $Z$ and $S$ (with $\gamma > 0$). {Thus, $\gamma$ is the relative secondary rate of spread of the malware. We consider all values of the parameter $\gamma$, with an associated trade-off: if $\gamma$ is high, the zombies spread too fast and increase visibility, while if $\gamma$ is low, the malware does not spread to cause significant damage.}

\vspace{-0.15in}
\subsection{SGZP Model with halting}\label{subsec:killing_model}

This model is akin to the previous one, with one more mechanism added: germinator nodes (G) can force a zombie (Z) to become passive (P) through a process we will call {\emph ``halting''}. This happens through another control variable $u_h$, which, in keeping with the intuition, can be thought of as the probability of halting encountered zombies at each instant. Again, we take $(u_Z, u_P, u_h)\in \mathcal{U}'$, where $\mathcal{U}'$ is the set of piecewise continuous controls. The system dynamics become: 
\vspace{-0.15in}
\begin{subequations}\label{model:killing}
\begin{align}
\dot{S} &=-\beta G S (u_{P}+ u_{Z}) -\gamma\beta Z S\\
\dot{Z} &= \beta G S u_{Z} + \gamma\beta Z S {- \pi \beta G Z u_h} \label{eq:Z_killing} \\
\dot{P} &=\beta G S u_{P} {+ \pi \beta G Z u_h},
\end{align}
\end{subequations}
with $0<\pi \leq 1$ signifying the extent to which the zombies can be stopped when encountered by the original germinators. {This model is similar to the Daley-Kendall rumor model \cite{daley1965stochastic}, where repeated interaction with active agents can turn an active spreader of the rumor into an agent that is aware of the rumor, but has no interest in spreading it any further.} 
The constraints now become:
\vspace{-0.1in}
\begin{subequations}\label{halt_bounds}
\begin{align}
u_P+ u_{Z} \leq 1\label{sum_bound_halt}\\
0 \leq u_P \leq 1, ~
0 \leq u_{Z}\leq 1,~
0\leq u_h\leq 1.\label{bound_halt}
\end{align}
\end{subequations}

\vspace{-0.15in}
\subsection{SGZP Model with no halting and adaptive defense}\label{subsec:beta_model}

Instead of allowing a constant rate of interactions $\beta$, the network defender can choose the effective mixing rate $\beta$ to be a function of the fraction of zombies as her defense policy ($\beta(Z)$). In these policies, the network defender regulates the rate of contact between nodes based on the proportion of zombie nodes it has observed. {While the network cannot determine which nodes have been compromised,} it can determine the fraction of the network that has been infected by zombies by observing the chatter among nodes and the extra communications whose purpose is unknown, either in the whole network or among a representative subset of nodes. {If these illicit communications are significant enough to attract the network defender's attention, they can implement a quarantine defense policy, captured by $\beta(Z)$, which will be a function of likelihood the malware is detected, and which will decrease the spread of the malware.}

We consider the system dynamics described in the no-halting model, and adapt them accordingly:

\vspace{-0.15in}
\begin{subequations}\label{model:betaZ}
\begin{align}
\dot{S} &=-\beta(Z) G S (u_P+ u_{Z}) -\gamma\beta(Z)  Z S\\
\dot{Z} &= \beta(Z) G S u_{Z} + \gamma\beta(Z) Z S \\
\dot{P} &=\beta(Z) G S u_P
\end{align}
\end{subequations}

The controls available are also the same as those in \eqref{nohalt_bounds}. In particular, they are still assumed to be piecewise continuity.

We consider two classes of $\beta(Z)$ functions: 1) Affine functions, of the form $\beta(Z) = - a Z + \beta_{max}$ for $0\leq a\leq \beta_{max}$ (a natural assumption, as the contact rate cannot be negative). If $a=0$, the affine case simplifies to the constant $\beta$ case. 2) Exponential sigmoids, of the form $\beta_Z =  \dfrac{\beta_0}{1+e^{ \alpha(Z - Z_{th})}}$, with $0<Z_{th}<1$ being a fixed threshold and $\alpha>0$ denoting the sharpness of the cut-off. As $\alpha$ increases, $\beta(Z)$ can become arbitrarily close to $\beta(Z)=\beta_0\mathbf{1}_{Z\leq Z_{th}}$, an all-or-nothing policy.
\hide{In the limit, the policies could be no action at all ($\beta(Z)=\beta_0$), which we examine in \S\ref{sec:constant},  or an all or nothing policy ($\beta(Z)=\beta_0\mathbf{1}_{Z\leq Z'}$ for some $Z'$ \red{flesh out in footnote}), which we will approach through sigmoid functions in \S\ref{sec:sigmoid}. In between, we will look at the intermediate case of an affine $\beta(Z)$ in \S\ref{sec:affine}.
We assume that} Both of these classes satisfy $\beta(Z)>0$ for all $Z$ (i.e., the network never shuts down completely due to the infection) and $\frac{d\beta(Z)}{dZ}<0$ for all $Z$ (except for the trivial case of constant $\beta(Z)$), as more visibility should lead to more communication restrictions from the network. In mobile epidemics, this is equivalent to nodes decreasing their communication range upon the detection of an infection, e.g. as in \cite{khouzani2011dynamic}. In practice, the network will have an estimate $\hat{Z}$ of the fraction of zombies. Our simulations reveal that the sub-optimality induced by the estimation error is small (\S \ref{sec:heuristics}).\hide{ The exact characterization of a stochastically optimal control for this instance is an interesting problem for future work.}
\vspace{-0.15in}
\subsection{Key observations}\label{subsec:observations}

We start with a theorem that holds for all the models presented above, and which will be used as a building block to obtain structural results in \S \ref{sec:structural}.

\begin{Theorem}\label{thm:constraints}
For a system with the mechanics described in either \S \ref{subsec:simple_model}, \S \ref{subsec:killing_model}, or \S \ref{subsec:beta_model}, with initial conditions $S(0)=S_0>0$, $G(0)=G_0>0$, $Z(0)= Z_0\geq0$, and $P(0)= P_0\geq0$, and $S_0 + G_0+Z_0+P_0= 1$, and with piecewise continuous controls $u_P$, $u_Z$
(and in \eqref{model:killing}, $u_h$), the dynamical systems~\eqref{model:nokilling}, \eqref{model:killing}, and \eqref{model:betaZ} have unique state solutions $(S(t), G(t), Z(t), P(t))$, with $S(t)>0$, $Z(t)\geq0$, $P(t)\geq0$, and $(S+G+Z+P)(t)=1$ for all $t \in [0,T]$. 
\end{Theorem}

The assumptions $S_0>0$  and $G_0 >0$ are natural, otherwise there is no interaction to control. Henceforth, we will assume these, as well as $Z_0\geq 0$ and $P_0\geq 0$.

{\begin{proof}
The uniqueness follows from standard results in the theory of ordinary differential equations~\cite[Theorem A.8, p. 419]{seierstad1987optimal}  given the observation that the RHS of the dynamic systems is comprised of quadratic forms and is thus Lipschitz over $[0,T]\times\mathbf{S}$, where $\mathbf{S}$ is the set of states such that the boundary conditions hold.

We provide the proof for the case of \S \ref{subsec:simple_model}, and note the changes for \S \ref{subsec:killing_model}. First of all, $(\dot{S}+\dot{Z}+\dot{P})(t)=0$ and $(S+Z+P)(0)=1-G_0$, so $(S+G+Z+P)(t)=1$ for all $t$. We know that $\dot{S}= -\beta G S (u_P+ u_{Z}) -\gamma\beta Z S \geq -M S$, where $M$ is the upperbound of $\beta G+\gamma\beta Z$ (because $(u_P+ u_{Z})\leq 1$). Therefore, $S(t)\geq S_0 e^{-Mt}>0$ for all $t$. Therefore, $\dot{Z}=\beta G S u_{Z} + \gamma\beta Z S \geq \gamma\beta Z S \geq M Z$, where $M$ is a lowerbound on $\gamma\beta S$ which exists due to continuity (respectively, $\dot{Z}=\beta G S u_{Z} + \gamma\beta Z S - \pi \beta Z G u_h \geq  Z (\gamma\beta S - \beta \pi G u_h )\geq M' Z$, where $M'$ is a lowerbound on $(\gamma\beta S - \beta \pi G u_h )$ which again exists due to continuity). Note that the first inequality resulted from $u_Z(t)\geq 0$ for all $t$. Therefore, $Z(t)\geq Z_0 e^{Mt}\geq 0$ (respectively $Z(t)\geq Z_0 e^{M't}\geq 0$) for all $t$. Finally, $\dot{P}=\beta G S u_P\geq 0$ for all $t$ (respectively, $\dot{P}=\beta G S u_P+ \pi \beta Z G u_h \geq 0$ for all $t$), as $u_Z(t)\geq 0$, so $P_0\geq 0$ leads to $P(t)\geq 0$ for all $t$.

Theorem \ref{thm:constraints} can be proved very similarly for the model in \S \ref{subsec:beta_model} using the reasoning we used for the model in \S \ref{subsec:simple_model}, with the difference that in the arguments, $\beta$ is replaced by $\beta(Z)$, which is lower-bounded away from zero for positive $Z$.
\end{proof}
}

\subsection{Utility Function}\label{subsec:cost}

As we discussed,  the attacker tries to maximize attack efficacy while minimizing visibility. We capture efficacy as a function $f(\cdot)$ of the aggregate number of zombies ($Z$) and passives ($P$) at each time instant. Meanwhile, visibility is only a function of zombies that re-spread the malware, {as that is the only time the malware is detectable. Visibility increases the likelihood that the network defender detects the malware and takes defensive actions.
}
 This means that we can capture instantaneous visibility as a function $g(\cdot)$ of the number of zombies at that instant. {While the attacker cannot in general measure the malware's visibility, she can choose $g(\cdot)$ based on how detrimental detection would be for her purposes.} This formulation is comprehensive because the fixed number of germinators ($G$) both cause damage and are visible, and are implicitly a term that is added to the variable of both functions.  This leads to the following aggregate damage function that the attacker seeks to maximize:
\begin{align}\label{objective}
J=\int_0^T (f(Z(t)+P(t))-g(Z(t)))\,dt.
\end{align}
We have some natural assumptions on $f(.)$ and $g(.)$:
$f(0)=g(0)=0$, with $\frac{dg(Z)}{dZ}>0$ and $\frac{\partial f(Z+P)}{\partial Z}=\frac{\partial f(Z+P)}{\partial P}>0$.

We assume that $f(x)$ is {\bf concave}, which means that incremental damage does not increase as the number of infected agents increases [i.e., the pay-off per infected agent decreases].

{\em In~\Sec\ref{sec:Optimal_simple}:} We assume $g(x)$ is {\bf convex}.  This means that an increment in the zombies is costlier (results in more visibility) when the infection is already more visible. This could be the case when the network becomes more wary of the infection as it progresses and becomes more visible.

{\em In~\Sec\ref{sec:Optimal_killing}:} We simplify $g$ to be {\bf linear}$,
g(x)=k_g x$, $k_g>0$.

 {\em In~\Sec\ref{sec:Optimal_beta}:} We set $g(x)\equiv0$, as the effects of visibility have been built into the network dynamics through $\beta(Z)$. This leaves us with:
 \begin{align}\label{objective_betaZ}
J=\int_0^T f(Z(t)+P(t))\,dt.
\end{align}

\vspace{-0.1in}
\subsection{Problem statement}\label{subsec:statement}
{\em In \S \ref{sec:Optimal_simple} and \S \ref{sec:Optimal_beta}, the attacker seeks to choose controls $(u_Z,u_P)\in \mathcal{U}$ satisfying \eqref{nohalt_bounds} so as to maximize $J$ (respectively, \eqref{objective} and \eqref{objective_betaZ}), while in \S \ref{sec:Optimal_killing}, she seeks to maximize $J$  \eqref{objective} through a choice of $(u_Z,u_P, u_h)\in \mathcal{U'}$ that satisfies \eqref{halt_bounds}.}

\section{Structural Results}\label{sec:structural}
Using Pontryagin's Maximum Principle and custom arguments specific to each case, we obtain the {\em one-jump bang-bang} structure of the optimal controls for the various cases in~\Sec\ref{subsec:simple_model}, \Sec\ref{subsec:killing_model}, and \Sec\ref{subsec:beta_model}. We provide the proof for \Sec\ref{sec:Optimal_simple} in the main text (\Sec\ref{sec:Optimal_simple1}) and the ones for \Sec\ref{sec:Optimal_killing} and \Sec\ref{sec:Optimal_beta} in the appendices (\S Appendix \ref{sec:Optimal_killing1} and \S Appendix \ref{sec:Optimal_beta1} respectively).

\hide{\subsection{Existence of optimal solutions}\label{subsec:existence}
\begin{Theorem}
For each of the problems outlined in \S \ref{sec:Sys_Model}, there exist piecewise continuous controls $(u^*_Z, u^*_P) \in \mathcal{U}$ (respectively $(u^*_Z, u^*_P, u^*_h) \in \mathcal{U}'$) and corresponding (unique) state solutions $(S(t),G(t), Z(t), P(t))$ to the initial value
problems such that $(u^*_Z, u^*_P) \in \arg\min_{(u_Z, u_P)\in  \mathcal{U}} J$(respectively, $(u^*_Z, u^*_P, u^*_h) \in \arg\min_{(u_Z, u_P, u_h)\in  \mathcal{U}} J$).
\end{Theorem}

\begin{proof}
This is a direct result of Cesari's theorem\cite[p. 60]{fleming1975deterministic}.

\end{proof}}

Intuition is unclear in determining these structures: while intuitively creating zombies at the beginning of the time period allows the malware to benefit from their epidemic spread, it also penalizes the malware more because of its prolonged visibility. This is further complicated by the fact that the controls can take any value between 0 and 1, and thus it is possible for the attacker to have any mix of malware spread at each instance in time. The strict structures that arise from the analysis are counter-intuitive and interesting both theoretically and from an implementation standpoint.

\vspace{-0.2in}
\subsection{Results for the no halting model (proved in \S \ref{sec:Optimal_simple1})}\label{sec:Optimal_simple}

\begin{Theorem}\label{thm:simple}
Any optimal control in $\mathcal{U}$ will satisfy
\begin{align*}
u_P(t)=
\begin{cases}
0\quad t \in[0,t^*)\\
1\quad  t \in(t^*,T)
\end{cases} u_Z(t)=
\begin{cases}
1\quad t \in[0,t^*)\\
0\quad   t \in(t^*,T)
\end{cases}
\end{align*}
for some $t^*\in [0,T)$.
\end{Theorem}

This result means that for any optimal control, there exists a time threshold $t^*$ such that prior to $t^*$, the germinators convert all the susceptibles they encounter to zombies, and subsequent to it they convert the susceptibles to passives.
 
The fact that creating zombies starts from the initial time for all interactions, that passives are created for a time period leading up to the terminal time for all interactions, and that the switch between creating zombies and passives is instantaneous -- with no gap between, and no over-lap in, the intervals in which these variants are propagated, as well as no intermediate propagation rates --  is not at all {\em a priori} obvious.

Note that we prove a {\em necessary} condition for any optimal control, thus reducing the search space of controls from a vector of functions to a scalar ($t^*$). This is a cause for concern, as the latter is much more computationally tractable for the attacker, and shows that any optimal policy will also be simple for the attacker to execute. {The attacker can execute the optimal policy by optimizing the ODE \eqref{model:nokilling}, just varying the scalar parameter $t^*$, and then coding $t^*$ into the germinators, which are the only nodes that execute the control.}

\vspace{-0.2in}
\subsection{Results for the halting model (proved in \S Appendix \ref{sec:Optimal_killing1})}\label{sec:Optimal_killing} 

\begin{Theorem}\label{thm:killing}
Any optimal control in $\mathcal{U'}$ will satisfy
\begin{align*}
u_P(t)=u_h(t)=
\begin{cases}
0~ t \in[0,t^*)\\
1~  t \in(t^*,T)
\end{cases} \hspace{-0.1in}u_Z(t)=
\begin{cases}
1~ t \in[0,t^*)\\
0~   t \in(t^*,T)
\end{cases}
\end{align*}
for some $t^*\in [0,T)$, except in the case where $Z(t)=0$ for all $t\in [0,T]$, in which case $u_h$ can be arbitrary with the other two structures holding.
\end{Theorem}

This means that there exists a time threshold $t^*$ such that prior to $t^*$, the germinators again convert all the susceptibles they encounter to zombies while not halting any zombies they meet, and subsequent to it they convert both the susceptibles and zombies they encounter to passives. Here, the added halting control can be used to slow the spread of zombies. 

The fact that the same result as Theorem \ref{thm:simple} holds for $u_Z$ and $u_P$ in the presence of $u_h$ is not clear {\em a priori}. Furthermore, the fact that the halting optimal control is bang-bang and that the switching time is the same as the other controls is surprising.

\vspace{-0.15in}
\subsection{Results for the adaptive defense model}\label{sec:Optimal_beta}

{\textit{Theorem \ref{thm:simple} holds (with the difference that $t^*\in [0,T]$) for constant, affine, and sigmoid $\beta(Z)$.
}} This is remarkable given that here, $\beta$ changes as a function of $Z$. This result is proved in \S Appendix \ref{sec:Optimal_beta1}.

\vspace{-0.1in}
\subsection{Proof of Theorem \ref{thm:simple} for the no halting model}\label{sec:Optimal_simple1}
\begin{proof}
 This proof utilizes the necessary conditions for an optimal control derived from Pontryagin's maximum principle. In particular, we explicitly characterize the optimal controls as functions of the optimal states and {\em co-states} (akin to Lagrange multipliers). Subsequently, we start at terminal time, where the co-states are known, and follow their evolution backward in time till we arrive at the initial time, thereby implicitly characterizing the necessary structure of the optimal controls.

  Define continuous co-states $(\lambda_S, \lambda_P, \lambda_Z, \lambda_0)$ such that at points of continuity of the controls: 
  \begin{align}\label{costates_nohalt}
\dot{\lambda}_S &= \beta[(\lambda_S-\lambda_P)  G u_P+ (\lambda_S-\lambda_Z)  G u_Z + (\lambda_S-\lambda_Z) \gamma Z]\notag\\
\dot{\lambda}_Z &= -f'(Z+P) + g'(Z)+(\lambda_S-\lambda_Z) \gamma\beta S\notag\\
\dot{\lambda}_P &=  -f'(Z+P),
\end{align}
with final co-state constraints:
\begin{align}\label{terminal_nohalt}
\lambda_S(T) = \lambda_Z(T) = \lambda_P(T) = 0.
\end{align}

Towards characterizing properties of optimal solutions, we define the {\em Hamiltonian} as:
\begin{align}\label{Hamiltonian_nohalt}
\ham (t) := \lambda_0(f(Z+P)-g(Z)) +(\lambda_P-\lambda_S) \beta G S u_P\notag \\+ (\lambda_Z-\lambda_S) \beta G S u_Z +
(\lambda_Z-\lambda_S) \gamma\beta Z S.
\end{align}
Pontryagin's Maximum Principle \cite[p.182]{seierstad1987optimal} states that any optimal control vector $u^*$ must satisfy the following necessary conditions:
\begin{flalign}
 &(\lambda_S, \lambda_P, \lambda_Z, \lambda_0)\neq \vec{0}, ~~~ \lambda_0 \in \{0,1\}\label{vec_neq_zero},
\\& \forall_{u \in \mathcal{U}, t\in[0,T]}~\ham({ S}^*, { Z}^*, { P}^*, u^*, \lambda_S(t), \lambda_P(t), \lambda_Z(t), \lambda_0, t) \geq\notag\\ &\ham({ S}^*, { Z}^*, { P}^*, u, \lambda_S(t), \lambda_P(t), \lambda_Z(t), \lambda_0, t)\label{maximization}.
\end{flalign}

But if $\lambda_0=0$, $(\lambda_S(T), \lambda_P(T), \lambda_Z(T), \lambda_0)= \vec{0}$, a contradiction, so $\lambda_0=1$.

\subsubsection{Structure of the optimal control}
If we define:
\begin{subequations}\label{phis_nohalt}
\begin{align}
\varphi_P = (\lambda_P-\lambda_S) \beta G S \label{phip_nohalt}\\
\varphi_Z = (\lambda_Z-\lambda_S) \beta G S\label{phiz_nohalt},
\end{align}
\end{subequations}
then, the Hamiltonian becomes:
\begin{align}
\ham (t) = f(Z+P)-g(Z) + \varphi_P u_P + \varphi_Z u_Z\notag\\ +
(\lambda_Z-\lambda_S) \gamma\beta Z S.
\end{align}

The maximization of the Hamiltonian \eqref{maximization}, added to the sum constraints for the controls \eqref{sum_bound_nohalt}, leads to the following optimality conditions for the controls:\footnote{The question marks (?) denote singular controls. These can occur when the coefficient of a control variable in the {\em augmented Hamiltonian} (which includes the constraints) is zero over an interval, and thus the control has no effect on the Hamiltonian maximizing condition of the PMP.}
\begin{subnumcases}{(u_P, u_Z)=\label{maxi_nohalt}}
(0,0)\qquad \varphi_P<0 ,~~\varphi_Z<0\\
(1,0)\qquad \varphi_P>0 ,~~\varphi_P > \varphi_Z \label{eq:b}\\
(0,1)\qquad \varphi_Z>0 ,~~\varphi_Z > \varphi_P\label{eq:c}\\
(?,?)\qquad \varphi_Z = \varphi_P\geq 0 \\
(?,0)\qquad \varphi_P = 0,~~\varphi_Z < 0 \\
(0,?)\qquad \varphi_Z = 0,~~\varphi_P < 0
\end{subnumcases}

From \eqref{phis_nohalt} and the state \eqref{model:nokilling} and costate \eqref{costates_nohalt} evolution equations and after some manipulations, we have:\footnote{$g'(Z):=\frac{dg(Z)}{dZ}, \quad f'(Z+P):=\frac{\partial f(Z+P)}{\partial Z}=\frac{\partial f(Z+P)}{\partial P}$}
\begin{subequations}\label{summary_nohalt}
\begin{align}
\dot{\varphi}_P &=\beta [G u_Z (\varphi_Z-\varphi_P)  + \gamma Z (\varphi_Z-\varphi_P)- G S f'(Z+P)] \notag\\
\dot{\varphi}_Z &=\beta [G S (g'(Z)-f'(Z+P)) \notag\\&\quad+ G u_P  (\varphi_P-\varphi_Z)- \gamma S \varphi_Z]\label{summary_nohalt_Z}\\
\dot{\varphi}_P& - \dot{\varphi}_Z ={- (\varphi_P-\varphi_Z) (\beta G u_Z + \gamma\beta Z + \beta G u_P)}\notag\\&\quad~~\qquad{- \beta G S g'(Z)+ \gamma\beta S \varphi_Z}\label{summary_nohalt_ZP}
,
\end{align}
\end{subequations}

{\subsubsection{Proof methodology outline}\label{subsec:methodology}
From here on, we will use the necessary optimality conditions to obtain timing conditions for phase transitions among the conditions in \eqref{maxi_nohalt}.
We show that a time $t^*$ exists such that, for $t\in (t^*,T)$, we have $u_P(t)=1$ and $u_Z(t)=0$ (\S \ref{subsec:T}). If $t^*=0$, we have finished characterizing optimal controls. If not (i.e., $t^*>0$), we prove that a time $t''$ exists such that for $t\in (t'',t^*)$, we have $u_P(t)=0$ and $u_Z(t)=1$ (in \S \ref{subsec:tstar}). Finally, we show that $t''$ must be equal to zero (in \S \ref{subsec:tdouble}), leading to all possible optimal controls agreeing with the structure laid out in Theorem \ref{thm:simple}. }

\subsubsection{Time interval leading up to $T$ and the existence of $t^*$}\label{subsec:T}
{We now follow the evolution of $\varphi_Z$ and $\varphi_P$ for a time interval leading to $T$ in order to characterize necessary conditions for the optimal controls and to prove the existence of $t^*$.} From the terminal time costate conditions \eqref{terminal_nohalt}:
\vspace{-0.1in}
\begin{align*}
&\varphi_P (T) = \varphi_Z (T) = 0,\\
&\dot{\varphi}_P (T^-)= -f'((Z+P)(T^-)) \beta G S(T^-) <0,\\
&\dot{\varphi}_P (T^-)- \dot{\varphi}_Z (T^-)= - \beta GS(T^-) g'(Z(T^-))<0.
\end{align*}
Therefore, {\bf  ${\varphi}_P (t) > max\{{\varphi}_Z (t), 0\}$} for some interval leading up to T due to the continuity of the states and costates and using the definition of a left derivative. Let $(t^*,T)$ be the largest interval over which  this holds for $t\in(t^*,T)$ for some $t^*<T$, leading to the fact that for all such $t$, $u_P(t)=1$ and $u_Z(t)=0$ due to \eqref{eq:b}.

For $t \in (t^*,T)$, \eqref{summary_nohalt} becomes:
\begin{subequations}\label{summary_nohalt_1}
\begin{align}
\dot{\varphi}_P &=-\beta G S f'(Z+P)+ \gamma\beta Z (\varphi_Z-\varphi_P)\label{summary_nohalt_1_P} \\
\dot{\varphi}_Z &=\beta G S (g'(Z)-f'(Z+P)) +\beta G (\varphi_P-\varphi_Z)- \gamma\beta S \varphi_Z\label{summary_nohalt_1_Z}\\
\dot{\varphi}_P& - \dot{\varphi}_Z ={ \gamma\beta S \varphi_Z- (\varphi_P-\varphi_Z) ( \gamma\beta Z + \beta G)}{- \beta G S g'(Z)}
.
\end{align}
\end{subequations}

Recall that $\varphi_P(t)>0$ for $t \in (t^*,T)$, so due to continuity, we either have  $\varphi_P(t^*)>0$ or $\varphi_P(t^*)=0$. We now rule out $\varphi_P(t^*)=0$. If $\varphi_P(t^*)=0$, Rolle's Mean Value Theorem \cite[p. 215]{stewart2006calculus} applies over the interval $(t^*,T)$: as  
$\varphi_P(t^*)=\varphi_P(T)=0$ and $\varphi_P$ is continuous and differentiable over this interval, there must exist $\tau\in (t^*,T)$ such that $\dot{\varphi}_P(\tau)=0$. However, from \eqref{summary_nohalt_1_P}, it can be seen that $\dot{\varphi}_P(t)<0$ for $t \in (t^*,T)$, a contradiction. Therefore, $\varphi_P(t^*)>0$.

 Thus, either $t^*=0$ or $\varphi_Z(t^*)=\varphi_P(t^*)$. If $t^*=0$, due to \eqref{eq:b}, we have $u_P(t)=1$ and $u_Z(t)=0$ for all $t$ which agrees with the structure in Theorem \ref{thm:simple}, so henceforth we focus on {\bf the case where $\varphi_Z(t^*)=\varphi_P(t^*)>0$}.

First, we derive a property that will prove useful later on. We have $\dot{Z}(t)\geq 0$ from \eqref{eq:Z_nokilling} and Theorem \ref{thm:constraints}, and thus due to the convexity of $g(\cdot)$ for $t<t^*$: 
\vspace{-0.05in}
\begin{align}\label{inequality_g}
\dfrac{G g'(Z(t^*))}{\gamma}\geq\dfrac{G g'(Z(t))}{\gamma}. 
\end{align}

\vspace{-0.05in}
Next, $Z(t^*)$ can either be equal to zero or strictly positive. We first show that if $Z(t^*)=0$, the structure holds. 

{{\bf If $Z(t^*)=0$}, we have $\dot{Z}= \gamma\beta S Z$
for $t\in(t^*,T)$ as $u_Z(t)=0$ in this interval. Consider $M_1>0$ to be an upper-bound on the continuous $\gamma\beta S$
in this interval, so we must have $Z(t)\leq Z(t^*)e^{M_1(t- t^*)}=0$, and therefore $Z(T)=0$ due to continuity and the uniqueness of solutions of first-order initial value problems.  Thus, as $\dot{Z}\geq 0$ for $t\in (0,T)$, we must have $\dot{Z}=0$ over this interval, which from \eqref{eq:Z_nokilling} and Theorem \ref{thm:constraints} leads to $u_Z(t)=0$ for $t\in(0,T)$ and $Z_0=0$. This also means that from \eqref{summary_nohalt_1_P}, $\dot{\varphi}_P (t) =-\beta G S f'(Z+P)<0$ in this interval, leading to $\varphi_P(t)> \varphi_P(T)=0$, and from \eqref{maxi_nohalt}, to $u_P(t)=1$ over this interval. Thus, again $t^*=0$, agreeing with the structure predicted by Theorem \ref{thm:simple}. So from now on we will {\bf consider $Z(t^*)>0$}.
}
\hide{$\dot{Z}= \beta G S u_Z + \gamma\beta Z S \geq  M_1 Z$ on $(0, t^*)$, where $M_1>0$ is an upper-bound on the continuous $\gamma\beta S$, so we must have $Z(t)\leq Z(t^*)e^{M_1(t- t^*)}=0$ due to an integral argument and thus $\dot{Z}=0$ over this interval, which from \eqref{eq:Z_nokilling} and Theorem \ref{thm:constraints} leads to $u_Z(t)=0$ for $t\in(0,t^*)$ and $Z_0=0$. This also means that from \eqref{summary_nohalt_1_P}, $\dot{\varphi}_P (t) =-\beta G S f'(Z+P)<0$ in this interval, leading to $\varphi_P(t)\geq \varphi_P(t^*)>0$ and therefore $u_P(t)=1$ over this interval, agreeing with the structure predicted by Theorem \ref{thm:simple}. Thus from now on we will {\bf consider $Z(t^*)>0$}.}

Now, we examine $ g'(Z(t^*))- f'((Z+P)(t^*))$, noting that it can either be positive or strictly negative, and investigate both cases in turn.

{\bf If $ g'(Z(t^*))- f'((Z+P)(t^*))\geq0$}, then $g'(Z(t))- f'((Z+P)(t))\geq0$ for all $t \in (t^*,T)$. This is because from \eqref{model:nokilling}, $\dot{P}(t)+\dot{Z}(t)\geq 0$ and $\dot{Z}(t)\geq 0$ over this interval, which coupled with the convexity of $g(\cdot)$ and $-f(\cdot)$ in their arguments gives the aforementioned result. From \eqref{summary_nohalt_1_Z} and the definition of $t^*$, $\dot{\varphi}_Z > - \gamma\beta S \varphi_Z \geq -M_2 \varphi_Z$ in this interval, with $M_2>0$ being an upper-bound on $\gamma\beta S$. Therefore, $\varphi_Z(t^*)\leq\varphi_Z(T)e^{-M_2(t^*-T)}=0$ due to an integral argument, which means that $\varphi_P (t^*) > 0 \geq {\varphi}_Z (t^*)$. Note that this would contradict the starting assumption of this segment, which was  $\varphi_P (t^*) = {\varphi}_Z (t^*)$

Therefore, from here on we will {\bf examine the case of $g'(Z(t^*))< f'((Z+P)(t^*))$}. 

\subsubsection{Time interval leading up to $t^*>0$ and the existence of $t''$}\label{subsec:tstar}
{We now look at the evolution of $\varphi_Z$ and $\varphi_P$ for a time interval leading to $t^*>0$, and show that $t''$ exists such that t for $t\in (t'',t^*)$, we have $u_P(t)=0$ and $u_Z(t)=1$. Furthermore, in these cases we showed $\varphi_Z(t^*)=\varphi_P(t^*)$, $Z(t^*)>0$, and $g'(Z(t^*))< f'((Z+P)(t^*))$.}
At such a point $t^*$, from \eqref{summary_nohalt_Z} and the continuity of the states and co-states:
\small
\begin{align}\label{eq:difference1}
(\dot{\varphi}_P (t^{*+})- \dot{\varphi}_Z(t^{*+}))=\beta S(t^{*}) [\gamma\varphi_Z(t^{*})- G g'(Z(t^*))].
\end{align}\normalsize
Now, \eqref{eq:difference1} should be positive, because if this derivative was strictly negative, the definition of the right-derivative would show that $\varphi_Z(t)>\varphi_P(t)$ for $t$ in an interval starting from $t^*$, a contradiction. Because from Theorem \ref{thm:constraints}, $S(t^*)>0$, so $ \beta S(t^{*}) [\gamma\varphi_Z(t^{*})- G g'(Z(t^*))]\geq0$ and:
\begin{align}\label{inequality_nohalt}
\varphi_Z(t^*) \geq \frac{ G g'(Z(t^*))}{\gamma}.
\end{align}
Now, we can see from a continuity argument on \eqref{summary_nohalt_Z} (given that $\varphi_Z(t^*)=\varphi_P(t^*)>0$) that $\dot{\varphi}_Z (t^{*-})<0$. Thus $\varphi_Z(t)>\varphi_Z(t^*)$ for some interval leading up to $t^*$ due to the definition of a left-derivative. 

From \eqref{summary_nohalt_Z}, \eqref{inequality_g}, and \eqref{inequality_nohalt}, we must have:
$\varphi_Z(t) > \dfrac{G g'(Z(t))}{\gamma} $ for $t$ in some interval leading up to $t^*$. Let $(t',t^*)$ be the maximal such interval. In this interval, from \eqref{summary_nohalt_ZP}, $\dot{\varphi}_P - \dot{\varphi}_Z >- (\varphi_P-\varphi_Z) ( \gamma\beta Z + \beta G) \geq - M_3 (\varphi_P-\varphi_Z)$, where $M_3>0$ is an upper-bound on the continuous expression $\gamma\beta Z + \beta G$. So for any $t$ in this interval, $(\varphi_P(t)-\varphi_Z(t))<(\varphi_P(t^*)-\varphi_Z(t^*))e^{-M_3(t-t^*)}=0$. Thus, ${\varphi}_P (t)<{\varphi}_Z (t)$ for $t \in (t',t^*)$. As $\varphi_Z(t^*)>0$, due to the continuity of the states and co-states, there exists a maximal interval $(t'',t^*)$ such that ${\varphi}_Z(t)>\max\{{\varphi}_P (t),0\}$. Following from \eqref{eq:c}, for $t\in(t'',t^*)$ we must have $u_P(t)=0$ and $u_Z(t)=1$.

\subsubsection{Proof that $t''=0$}\label{subsec:tdouble}
 If $t'' =0$, the above concludes our specification of the structure, which agrees with Theorem \ref{thm:simple}. Thus, henceforth we assume $t''>0$, and thus {\bf either $\varphi_Z(t'')={\varphi}_P (t'')$ or ${\varphi}_Z(t'')=0$}.

For $t \in (t{''}, t^*)$, \eqref{summary_nohalt} becomes:
\begin{subequations}
\begin{align}
\dot{\varphi}_P &=\beta [ -G S f'(Z+P)+ G  (\varphi_Z-\varphi_P)  + \gamma Z (\varphi_Z-\varphi_P)] \\
\dot{\varphi}_Z &=\beta [ G S (g'(Z)-f'(Z+P)) - \gamma S \varphi_Z]\label{eq:simple_Z2}\\
\dot{\varphi}_P& - \dot{\varphi}_Z =\beta [{ \gamma S \varphi_Z- (\varphi_P-\varphi_Z) ( G + \gamma Z)- G S g'(Z)}]
,
\end{align}
\end{subequations}
Now, for $t\in(t'',t^*)$, $g'(Z(t))-f'((Z+P)(t)) < g'(Z(t^*))-f'((Z+P)(t^*))<0$. This is because $\dot{Z}(t)>0$ as $u_Z(t)=1$, and $\dot{P}(t)=0$ as $u_P(t)=0$, so $g(\cdot)-f(\cdot)$ is convex in the strictly increasing $Z$ in this interval. So from \eqref{eq:simple_Z2}, $\dot{\varphi}_Z <- \gamma\beta S \varphi_Z\leq -M_4 \varphi_Z$ with $M_4>0$ being the upper-bound of the continuous $\gamma\beta S$, and therefore for all $t\in (t'',t^*)$, 
$\varphi_Z(t)\geq\varphi_Z(t^*)e^{-M_4(t-t^*)}$, and therefore by continuity, $\varphi_Z(t'')\geq\varphi_Z(t^*)e^{-M_4(t''-t^*)}$. Thus, we can conclude that $\varphi_Z(t{''})>0$, as $\varphi_Z(t^*)>0$.

 So for $t''>0$, {\bf we must have ${\varphi}_P (t{''}) = {\varphi}_Z (t{''})$}. In this case, we have
 $(\dot{\varphi}_P  (t''^{+}) - \dot{\varphi}_Z (t''^{+}))\leq 0$, as if it is strictly positive, an integral argument will lead to a contradiction with ${\varphi}_P (t) <{\varphi}_Z (t)$ for $t \in (t{''}, t^*)$. Using the continuity of the states and co-states and as from Theorem \ref{thm:constraints}, $S(t'')>0$, \eqref{eq:simple_Z2} becomes:
\begin{align}
\dot{\varphi}_P  (t''^{+}) - \dot{\varphi}_Z (t''^{+}) &= \beta S(t'') [\gamma\varphi_Z(t'')- G g'(Z(t''))]\leq 0\notag
\\&\hence \quad  \varphi_Z(t{''}) \leq \frac{G g'(Z(t''))}{\gamma},\label{inequality2_nohalt}
\end{align}

We know that for all $t \in (t{''}, t^*)$, $g'(Z(t))-f'(Z+P(t)) <0$, so from \eqref{eq:simple_Z2}, $\dot{\varphi}_Z(t)<-\gamma\beta S\varphi_Z < -M_5 \varphi_Z<0$, where $M_5>0$ is an upper-bound on the continuous $\gamma\beta S$. Thus,
\begin{align}\label{inequality_Z}
\varphi_Z(t'')>\varphi_Z(t^*).
\end{align}
But \eqref{inequality_g}, \eqref{inequality2_nohalt}, and \eqref{inequality_Z} lead to $\varphi_Z(t^*) < \dfrac{ G g'(Z(t^*))}{\gamma}$, which contradicts \eqref{inequality_nohalt}.

Thus $t{''}=0$, and this concludes our specification of the structure of the optimal controls which conform to the structure set out in Theorem \ref{thm:simple}. 
\end{proof}

\section{Simulation}\label{sec:heuristics}
In the preceding sections, we showed that the optimal spreading controls of the malware in all of the described settings can be fully described by a scalar parameter $t^*$. In this section, we investigate the variation of $t^*$ with respect to some system parameters and then compare the relative performance of the optimal spreading controls with simple heuristics (\S \ref{subsec:fig1}).\footnote{{Stealth conscious epidemics are an emerging threat, and while more data is available now than before, their very nature makes real spreading
data hard to come by and a topic of active research, even years after the fact. Thus, our numerical studies are based on simulations with parameters that are justified based on their real-world implication.}}
In these studies, the main parameter of variation is $\gamma$, as a higher $\gamma$ indicates that zombies spread at a faster rate than infection via germination, and thus $\gamma$ represents a measure of the virility of the zombie malware variant. Varying $\gamma$ changes the relative contact rates internal to the model and thus represents different possible dynamics of a malware attack. In contrast, varying $\beta$, the contact rate of germinators and susceptibles, changes the number of contacts across the board, which is equivalent to changing $T$. Thus any variation of $\beta$ would only show how $t^*$ changes for a specific epidemic. \hide{We have studied this variation in the technical report \cite{eshghi2015visibility}.} Finally, we numerically investigate the fragility of the optimal control to network estimation errors in the adaptive defense model and to synchronization errors among germinators (\S \ref{subsec:fig3}).

\subsection{Structure of the optimal malware spread controls and their performance vs heuristics}\label{subsec:fig1}
{We first computed $t^*$ (the optimal switching time) as a function of the relative spread rate of the zombies $\gamma$ for the problems in \S \ref{subsec:simple_model} and \S \ref{subsec:killing_model} (with different values of halting efficacy $\pi$), as well as the optimal controls, for a cost function for which both Theorem \ref{thm:simple} and \ref{thm:killing} apply (Figure \ref{fig:one}). As $\gamma$ increases, zombies are created for a shorter period due to the rapid explosion of their population later on. Furthermore, the addition of a halting control and its increased efficacy leads to the attacker creating zombies for longer, as she can control their spread (and thus their visibility) later on using the halting control. } 

\begin{figure}[htbp]
\centering
\includegraphics[scale=0.24]{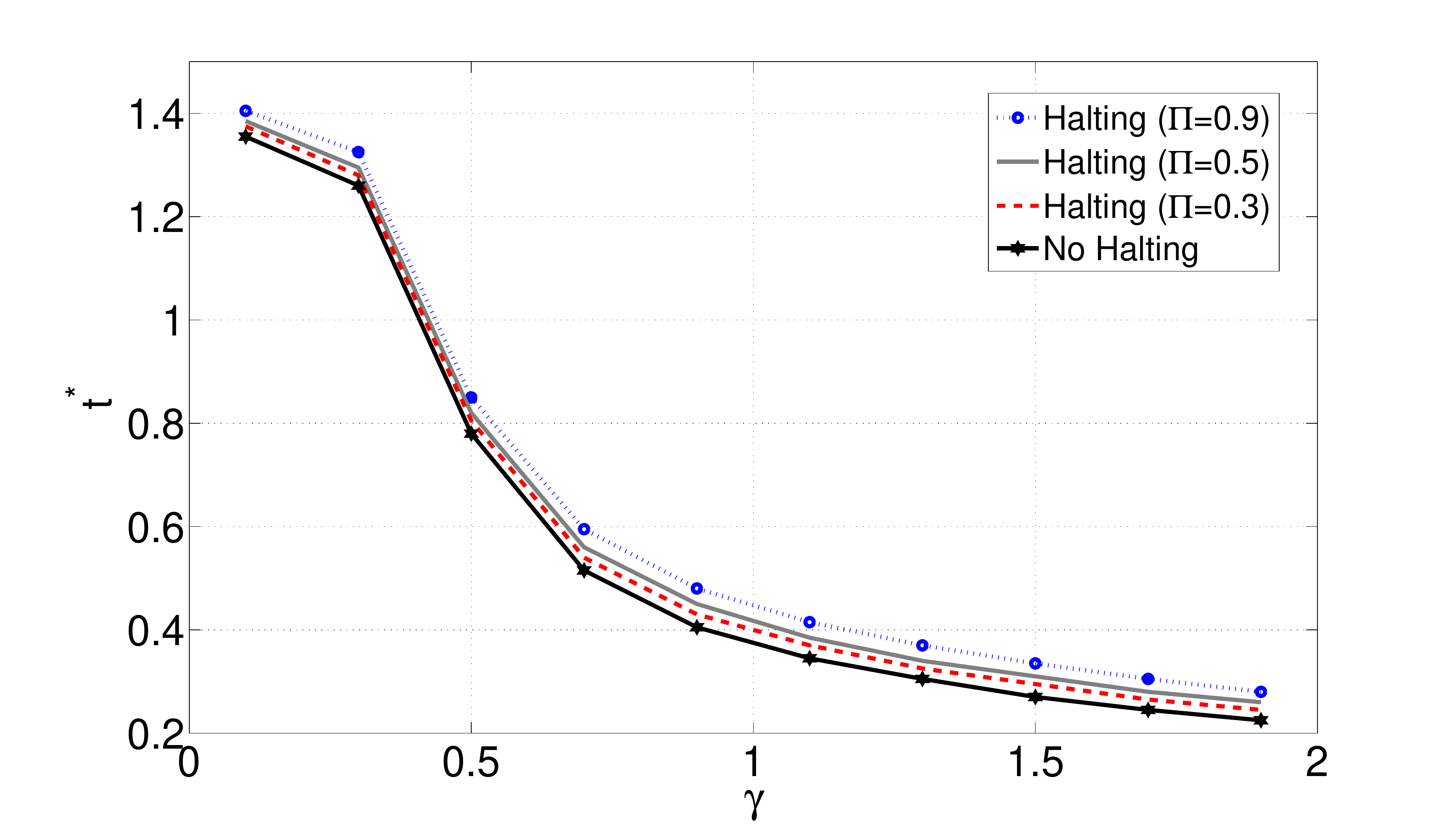}
\caption{We compared $t^*$ (the length of time the zombie control $u_Z$ was equal to one) for the optiml no halting and halting controls as the secondary rate of spread of the zombies ($\gamma$) was varied. Here, $\beta=2$, $T=5$,$(S_0,G_0,Z_0,P_0)=(0.99, 0.01,0,0)$, $f(x)= x^{0.5}$, and $g(x)= k_g x = 0.7 x$.}\label{fig:one}
\end{figure}

{We then compared the cost of these two optimal controls to that of simple heuristics: for the model in \S \ref{subsec:simple_model}, {\em Always Zombie} and {\em Always Passive} represent the two most extreme policies - {\em Always Zombie} sets $u_Z(t)=1$ and $u_P(t)=0$ for all times, while {\em Always Passive} does the exact opposite.
Thus, in these heuristics the germinators only ever propagate one fixed type of malware variant. In the {\em Optimal Static Mixing} heuristic, the attacker chooses a fixed ratio for $u_Z$ and $u_P$ at all times. Our optimal controls are titled {\em No Halting} and {\em Halting}, the latter indexed by the value of $\pi$ (which represents the relative success of the germinators in halting zombies). The efficacy of the policies is evaluated as $\gamma$, the relative propagation rate of the zombies is varied (Figure \ref{fig:two}, which is presented for the same parameters as those used in Figure \ref{fig:one}).

The optimal controls perform much better than the heuristics, with the halting control outperforming the no-halting control for by as much as 10\% for large values of $\pi$ (where the halting control is efficient) and $\gamma$ (where the zombie variant propagation is rapid), both factors which penalize sub-optimal decision-making. This vindicates the assumption that the attacker would be wise to utilize the halting control were it to be available. Out of the simple heuristics, optimal static mixing has the maximum utility, which is typically 10\% below that of even the no-halting optimal control.}

\begin{figure}[htbp]
\centering
\includegraphics[scale=0.24]{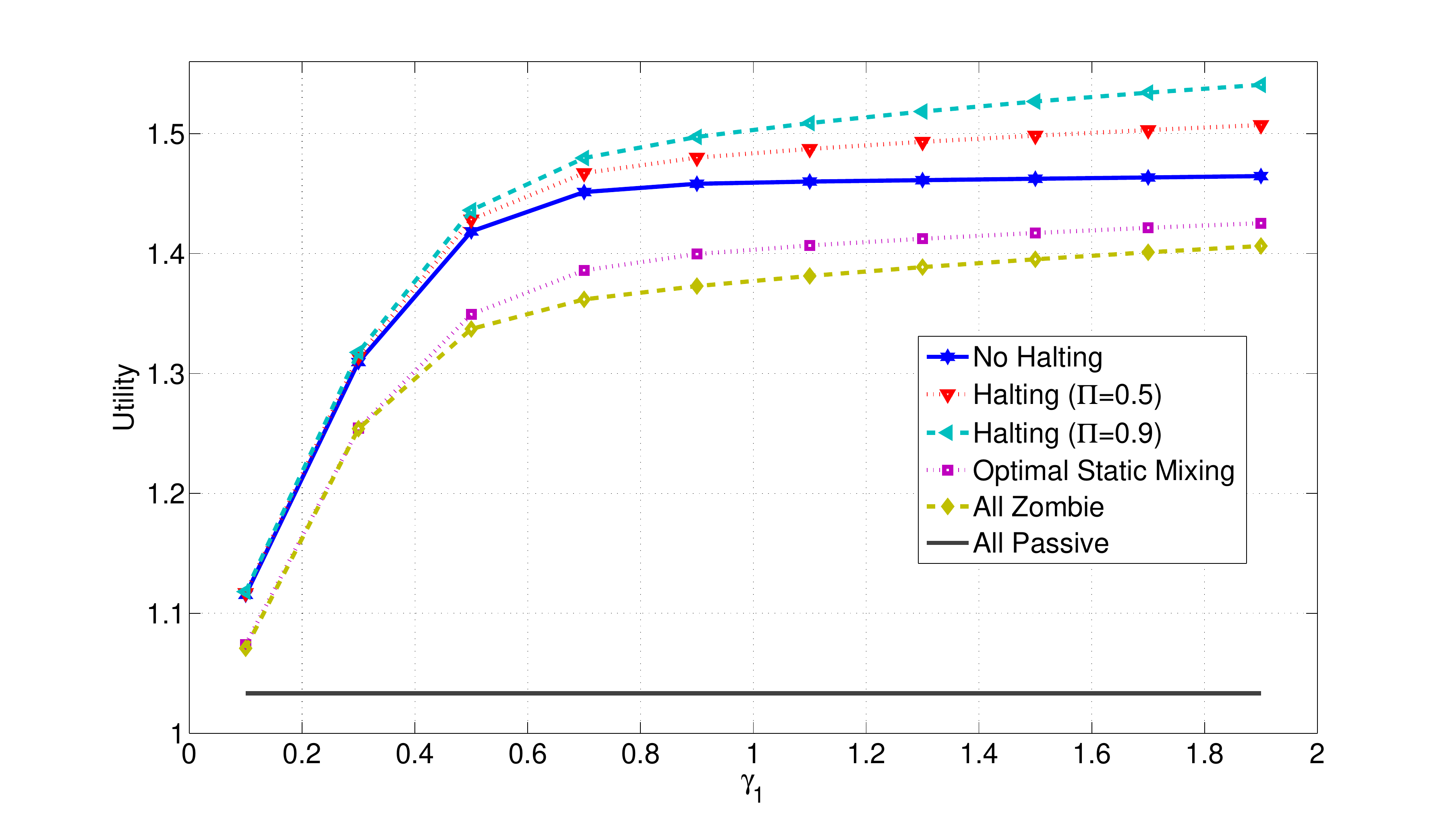}
\caption{Comparison of the damage utilities across the optimal controls and heuristics for the parameters of Fig. \ref{fig:one}.}\label{fig:two}
\end{figure}

\subsection{Fragility of the optimal damage to network estimation errors and synchronization errors in the germinators}\label{subsec:fig3}
We then investigated how the optimal control would fare when the network, which is capable of adaptive defense (i.e., the model in \S\ref{subsec:beta_model}), has an erroneous estimate of the fraction of zombies (Figure \ref{fig:three}). The optimal attack policy is derived with the assumption that the network's defense policy is based on the correct observation of the visibility of the epidemic (i.e., the fraction of zombies), information that is rarely available. Figure \ref{fig:three} shows that the optimal control is remarkably robust to the network's estimation errors up, with an error of 5\% even when the estimation error is 40\%. In many cases, the performance is much better.

\begin{figure}[htbp]
\centering
\includegraphics[scale=0.24]{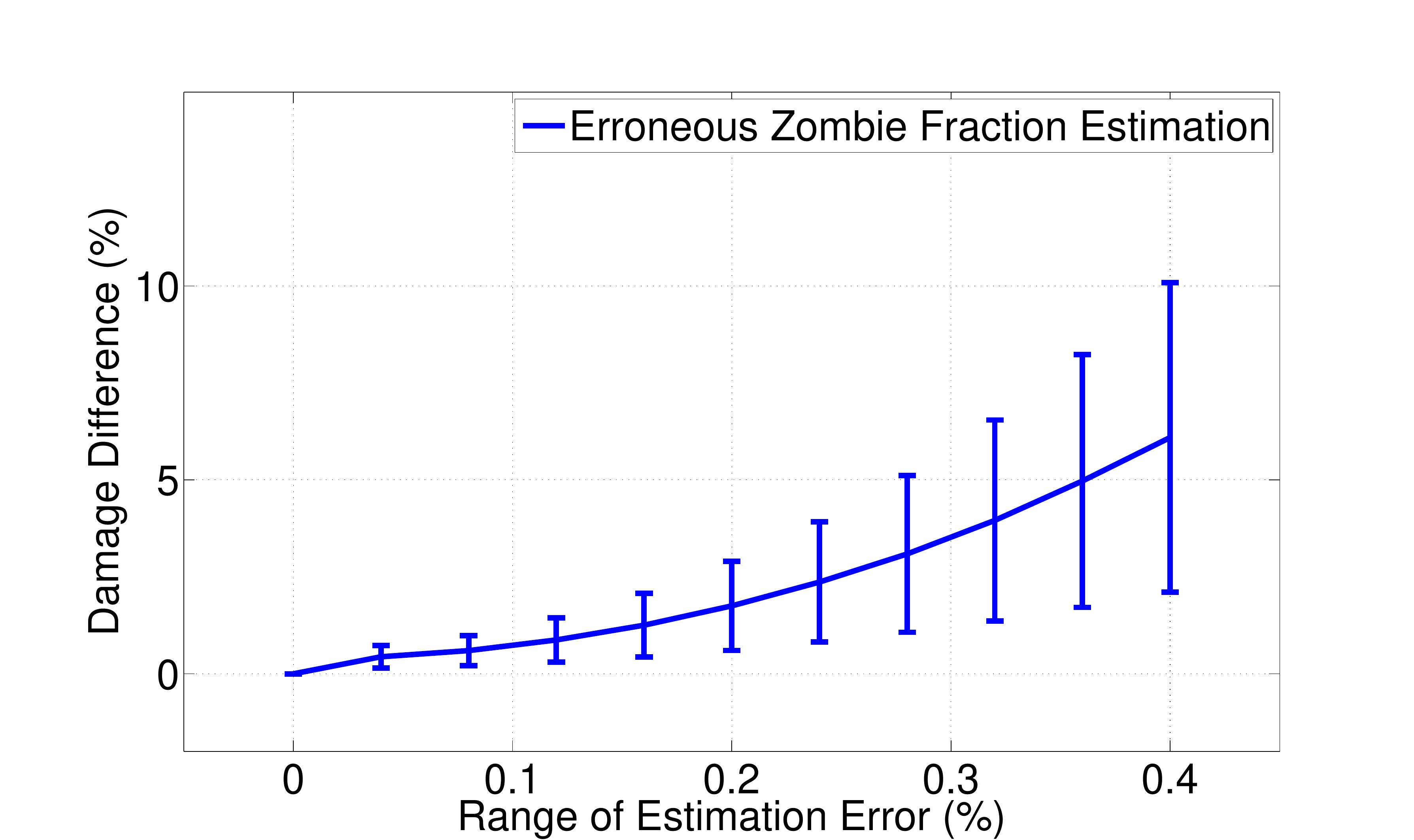}
\caption{The network was assumed to make unbiased random estimation errors at each time instant with the range depicted on the x-axis. The solid line shows the average difference in damage relative to the optimal over 50 runs of the estimating network. Here, we used an exponential sigmoid $\beta(Z)$ with $\beta_0=1$, $\alpha=100$, $T=15$, $\gamma=1.4$, $Z_{th}=0.01$, $(S_0,G_0,Z_0,P_0)=(0.999, 0.001,0,0)$, and $f(x)=x^{0.9}$.}\label{fig:three}
\end{figure}

Finally, we examined how synchronization errors among the germinators would affect the utility of the malware. One of the benefits of the malware spread models was that they assumed that only this small fraction of nodes, which is under the direct control of the attacker, has to coordinate their actions. To examine the fragility of the optimal control to this coordination, once the optimal policy is derived, random errors are introduced to the clocks of the germinators, and the resulting utilities are compared over 100 runs of the simulation (Figure \ref{fig:four}). As can be seen, the damage of both the no halting ($\pi=0$) and halting ($\pi=0.5$) cases is distributed around the damage obtained by the calculated optimal control, and only suffers a 10-15\% performance drop for synchronization errors of up to $30\%$ of $t^*$ in the small number of germinators.

Furthermore, it can be seen that the synchronized infinite-node optimal control can actually perform slightly worse than the case where there are synchronization errors on a finite number of nodes, even in the mean. We can explain this as follows: in the previous sections, we characterized the optimal solution for the problem in \S \ref{subsec:statement} under the assumption that the number of nodes was infinite. For a finite number of nodes, even without synchronization errors, the damage sustained by the simulated network can be different from (and potentially less than) that computed using the computational optimal control framework. 

These studies lead to the conclusion that an adversary will not be deterred by the possibility of errors in estimation and synchronization of the malware spread, further sounding the alarm about the emerging trend of visibility-aware malware.

\begin{figure}[htbp]
\centering
\includegraphics[scale=0.24]{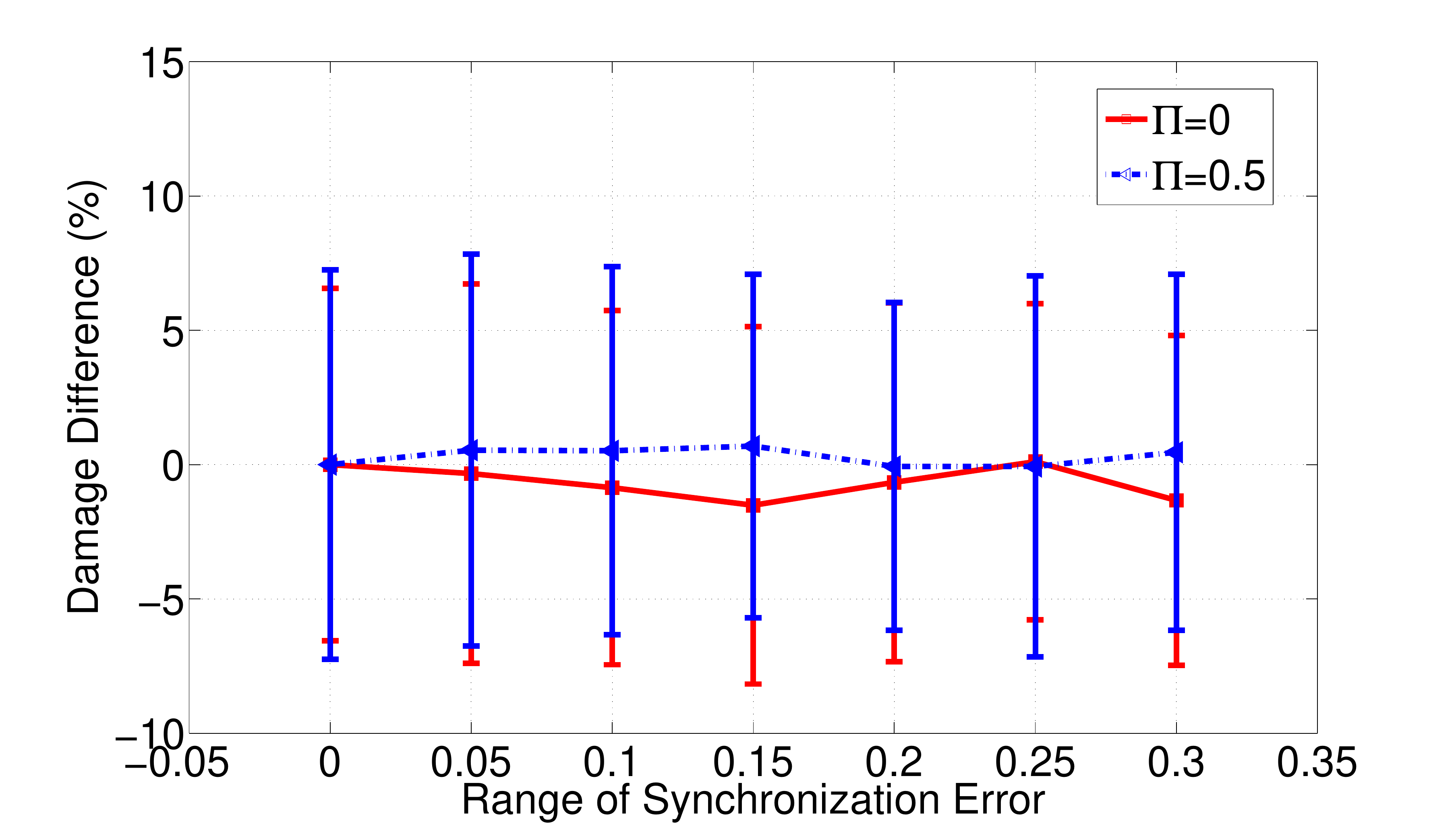}
\caption{Germinators were assumed to have unbiased random synchronization errors at each time instant with the range depicted on the x-axis. The lines shows average damage over 100 runs with unsynchronized germinators. Here, $\beta=2$, $\gamma=0.5$, $T=5$, $(S_0,G_0,Z_0,P_0)=(0.99, 0.01,0,0)$, $f(x)= x^{0.5}$, and $g(x)= k_g x = 0.7 x$, and the simulation was run for 500 nodes (i.e., 5 germinators).}\label{fig:four}
\end{figure}

\hide{We also look at visibility and damage separately on a network trace from SigComm '09 in Figure \ref{fig:4}. This data-set includes ``Bluetooth encounters, opportunistic messaging, and social profiles of 76 users of MobiClique application at SIGCOMM 2009''. We show how the damage and visibility of an infection evolves using the optimal policy, as well as  the case were sub-optimal heuristics are used.}

\vspace{-0.1in}

{\section{Future Directions}
{In this paper we investigated the optimal controls for the SGZP model with and without halting with no explicit network defense (\S \ref{subsec:simple_model} and \ref{subsec:killing_model}), and without halting for the case with adaptive network defense (\S \ref{subsec:beta_model}). This leaves open the case of the SGZP model with halting and adaptive defense. Initial analytical investigations show that Theorem \ref{thm:killing} is likely generalizable to this case, barring some technical issues that will be investigated in the future. {In principle, $\gamma$ can also be a variable to be optimized by the attacker in all models. Furthermore, the model can be extended to a botnet case where the attack is unleashed only when the damage-visibility trade-off is at the optimal point -- the same arguments as in the paper would hold in that case, with the difference that the terminal time will be free.} The set-up and formulation of the visibility problem is, to the best of our knowledge, novel, and thus leads itself to analysis both in the mean-field regime and in more structured settings. In particular, in the mean-field case, possible patching will be addressed at a later stage, as well as the dynamic game that would result from such a competition. 

{The current work is an abstraction of practical cybersecurity problems mainly due to the homogeneous mixing assumption.} Another possible direction is to look at the optimal control of such an epidemic in sub-populations with differentiating characteristics (e.g., location, contact rate) {as a way to relax the homogeneous mixing assumption (e.g., by following the roadmap in \cite{eshghi2015optimal})}. Such a generalization would better model Stuxnet in particular, with the goal being to maximize the number of infected agents in a particular region, while minimizing the total number of detectable zombies.}

%%%%%%%%%%%%%%%%%%%%%%%%%%%%%%%%%%%%%%%%%%%%%%%%%%%%%%%%%%%%%%%%%%%%%%%%%%%%%%%%

\vspace{-0.2in}

\bibliographystyle{ieeetr}
\bibliography{Reference1}

\begin{appendices}

\vspace{-0.1in}
\section{Proof of Theorem \ref{thm:killing}}\label{sec:Optimal_killing1}

\begin{proof}
This proof follows the same structure as that of Theorem \ref{thm:simple}.

As before, we define continuous co-states $(\lambda_S, \lambda_P, \lambda_Z, \lambda_0)$ such that at points of continuity of the controls:   

\begin{align}\label{costates}
\dot{\lambda}_S &= (\lambda_S-\lambda_P) \beta G u_P+ (\lambda_S-\lambda_Z)[\beta G u_Z +  \gamma \beta Z]\notag\\
\dot{\lambda}_Z &=  \lambda_0g'(Z)-\lambda_0f'(Z+P) +(\lambda_S-\lambda_Z) \gamma \beta S \notag\\&~~+ (\lambda_Z-\lambda_P) \pi \beta G u_{h}\notag\\
\dot{\lambda}_P &=  -\lambda_0f'(Z+P),
\end{align}
with final state constraints:
\begin{align}\label{terminal}
\lambda_S(T) = \lambda_Z(T) = \lambda_P(T) = 0.
\end{align}
To characterize optimal controls, we define the Hamiltonian to be:
\begin{align}
\ham (t) = & \lambda_0 (f(Z+P)-g(Z)) +(\lambda_P-\lambda_Z) \pi \beta G Z u_{h}\notag \\&\hspace{-0.5in}+ (\lambda_Z-\lambda_S) [\beta G S u_Z+
\gamma \beta Z S ]+(\lambda_P-\lambda_S) \beta G S u_P.
\end{align}
Pontryagin's Maximum Principle again gives the following necessary conditions for an optimal control vector $u^*$:
\begin{flalign}
 &(\lambda_S, \lambda_P, \lambda_Z, \lambda_0)\neq \vec{0}~~~ \lambda_0\in\{0,1\}\label{vec_neq_zero1},
\\& \forall_{u \in \mathcal{U}, t\in[0,T]}~\ham({ S}^*, { Z}^*, { P}^*, u^*, \lambda_S(t), \lambda_P(t), \lambda_Z(t), \lambda_0, t) \geq\notag\\ &\ham({ S}^*, { Z}^*, { P}^*, u, \lambda_S(t), \lambda_P(t), \lambda_Z(t), \lambda_0, t)\label{maximization1}.
\end{flalign}

Again, if $\lambda_0=0$, $(\lambda_S(T), \lambda_P(T), \lambda_Z(T), \lambda_0)= \vec{0}$, a contradiction, so $\lambda_0=1$.

Now, we have:
\begin{align*}
\dot{\lambda}_P-\dot{\lambda}_Z &= - g'(Z) -(\lambda_S-\lambda_Z) \gamma \beta S - (\lambda_Z-\lambda_P) \pi \beta G u_{h} \\
\dot{\lambda}_S-\dot{\lambda}_Z &= f'(Z+P)- g'(Z)+ (\lambda_S-\lambda_P) \beta G u_P\\&~~+ (\lambda_S-\lambda_Z) \beta G u_Z + (\lambda_S-\lambda_Z) \gamma \beta (Z-S)\\&~~ - (\lambda_Z-\lambda_P) \pi \beta G u_{h}\\
\dot{\lambda}_S -\dot{\lambda}_P &= f'(Z+P)+ (\lambda_S-\lambda_Z) [\beta G u_Z+ \gamma \beta Z]\\&~~+ (\lambda_S-\lambda_P) \beta G u_P,
\end{align*}

\subsubsection{Structure of the optimal control}
If we define:
\begin{subequations}\label{phis}
\begin{align}
\varphi_P &= (\lambda_P-\lambda_S) \beta G S \\
\varphi_Z &= (\lambda_Z-\lambda_S) \beta G S\\
\varphi_{h} &= (\lambda_P-\lambda_Z) \pi \beta G Z,
\end{align}
\end{subequations}
then, the Hamiltonian becomes:
\begin{align*}\small
\ham (t) &= f(Z+P)-g(Z) + \varphi_P u_P + \varphi_Z u_Z + \varphi_{h} u_{h} \notag\\&~~+
(\lambda_Z-\lambda_S) \gamma \beta Z S.
\end{align*} 
Also notice that:
\begin{align}\label{phicondition}
\varphi_{h}  = \pi \dfrac{Z}{S} (\varphi_P-\varphi_Z).
\end{align}

The maximization of the Hamiltonian \eqref{maximization1}, added to the sum constraints for the controls \eqref{sum_bound_nohalt}, leads to the following optimality conditions for the controls:

\vspace{-0.1in}
\begin{subnumcases}{(u_P, u_Z)=}\label{maxi}
(0,0)\qquad \varphi_P<0 ,~~\varphi_Z<0\\
(1,0)\qquad \varphi_P>0 ,~~\varphi_P > \varphi_Z \label{eq:b2}\\
(0,1)\qquad \varphi_Z>0 ,~~\varphi_Z > \varphi_P\label{eq:c2}\\
(?,?)\qquad \varphi_Z = \varphi_P\geq 0 \\
(?,0)\qquad \varphi_P = 0,~~\varphi_Z < 0 \\
(0,?)\qquad \varphi_Z = 0,~~\varphi_P < 0
\end{subnumcases}
Furthermore,
\begin{align}\label{sum_halting}
\varphi_Z(t)>0~ \text{or}~\varphi_P(t)>0 \hence u_P(t) + u_Z(t) =1,
\end{align} 
as if that is not true, we can increase $\ham(t)$ by adding to either $u_P(t)$ or $u_Z(t)$, a contradiction with the Hamiltonian maximization condition of the Maximum Principle \eqref{maximization1}. Also, 
\begin{subnumcases}{u_{h}=}
0\qquad \varphi_{h}<0 \\
1\qquad \varphi_{h}>0 \label{maxi1}\\
?\qquad \varphi_{h}=0
\end{subnumcases}.
Using \eqref{phicondition}, we can rewrite the above as:
\begin{subnumcases}{u_{h}=}\label{maxi2}
0\qquad \varphi_{P}<\varphi_Z ~\&~ Z(t)>0\label{eq:h_cases_0}\\
1\qquad \varphi_{P}>\varphi_Z~ \&~ Z(t)>0\label{eq:h_cases_1}\\
?\qquad \varphi_{P}=\varphi_Z~ \text{or}~Z(t)=0\label{eq:h_cases_2}
\end{subnumcases}.
From \eqref{phis} and the state and costate evolution equations and after trite manipulation, we have:
\begin{subequations}\label{summary_halt}
\begin{align}
\dot{\varphi}_P &=-\beta G S f'(Z+P)+\beta G u_Z (\varphi_Z-\varphi_P)  \notag\\&~~+ \gamma \beta Z (\varphi_Z-\varphi_P)\label{eq:P_dot} \\
\dot{\varphi}_Z &=\beta G S (k_g-f'(Z+P)) - \gamma \beta S \varphi_Z\notag\\&~~+\beta G (u_P- \pi u_{h})  (\varphi_P-\varphi_Z) \\
\dot{\varphi}_P - \dot{\varphi}_Z &=- (\varphi_P-\varphi_Z) (\beta G u_Z + \gamma \beta Z + \beta G u_P - \beta G u_{h})\notag\\&~~- \beta G S k_g+ \gamma \beta S \varphi_Z\\
\dot{\varphi}_{h}
&= - \pi \beta G Z k_g  + \pi \beta G u_Z (\varphi_P-\varphi_Z) + \pi \gamma \beta Z \varphi_{P}.
\end{align}
\end{subequations}
{From here on, the proof follows the same outline laid out in \S \ref{subsec:methodology} (in terms of finding $t^*$ and $t''$ and proving $t''=0$); however, the algebraic expressions for $\dot{\varphi}_Z$, $\dot{\varphi}_P$ are different and $\varphi_h(t)$ is introduced in the dynamics, necessitating the use of different and context-specific analytical arguments.}

\subsubsection{Time interval leading up to $T$ and the existence of $t^*$}
{We follow the evolution of $\varphi_Z$, $\varphi_P$, and $\varphi_h$ for a time interval leading to $T$ and prove the existence of $t^*$ such that we have $u_P(t)=1$, $u_Z(t)=0$, and, if $Z(T)>0$, $u_h(t)=1$ for all $t\in(t^*,T)$ (otherwise, $u_h$ can be arbitrary over this interval).} From the terminal time costate conditions \eqref{terminal}:
\begin{subequations}
\begin{align}
&\varphi_P (T) = \varphi_Z (T) =\varphi_{h} (T) = 0,\label{eq:terminals}\\
&\dot{\varphi}_P (T^-)= -f'((Z+P)(T^-)) \beta G S(T^-) < 0,\\
&\dot{\varphi}_P (T^-)- \dot{\varphi}_Z (T^-)= - \beta GS(T^-) k_g< 0,\label{eq:terminal_dots}\\
&\dot{\varphi}_{h} (T^-)= - \pi \beta G Z(T^-) k_g \leq 0. 
\end{align}
\end{subequations}

Now, we may either have $Z(T)=0$ or $Z(T)>0$ due to Theorem \ref{thm:constraints}.

We start by {\bf considering the case where $Z(T)=0$}. From \eqref{eq:Z_killing} we have $\dot{Z} \geq Z (\gamma \beta S {- \pi \beta G u_{h}})\geq M_6 Z$ for $t \in [0, T]$, where $M_6>0$ is an upper-bound on the $\gamma \beta S$ over the whole interval. Therefore, $Z(t)e^{M_6(t-T)}\leq Z(T)=0$. Thus we must have $Z(t)=0$ for all $t\in[0,T]$. This means that $\dot{Z}(t)= \beta G S u_Z =0$ over this interval, which from Theorem \ref{thm:constraints} leads to $u_Z(t)=0$ for all $t\in [0,T]$. Furthermore, as $Z(t)$ is never positive, $u_h(t)$ will have no effect on the dynamics of the system, and can thus be arbitrary. Finally, \eqref{eq:P_dot} and \eqref{eq:terminals} tell us that $\varphi_P(T)=0$ and $\dot{\varphi}_P(t)=-\beta G S f'(P)<0$ over this interval, which leads to $\varphi_P(t)>0$ for $t\in [0,T)$ due to continuity of the states and co-states and the differentiability of $\varphi_P(t)$ using an integral argument. This, along with $u_Z(t)=0$ for all $t\in [0,T]$ and \eqref{sum_halting}
 leads to $u_P(t)=1$ for all $t\in [0,T)$ {(and therefore $t^*=0$)}. So in sum, for all $t\in [0,T)$, $u_P(t)=1$, $u_Z(t)=0$, with $u_h(t)$ taking any arbitrary value. This agrees with the structure set forth in Theorem \ref{thm:killing}.

Henceforth, {\bf we examine the case where $Z(T)>0$}. From \eqref{eq:terminals} and \eqref{eq:terminal_dots}, as before, ${\varphi}_P (t) > max\{{\varphi}_Z (t), 0\}$ for some interval leading up to T due to the continuity of the states and costates and using the definition of a left derivative. Let $(t^*,T)$ be the largest interval over which  this holds for $t\in(t^*,T)$ for some $t^*<T$, leading to the fact that for all such $t$, $u_P(t)=1$ and $u_Z(t)=0$ due to \eqref{eq:b2}.

We now prove that for $t\in [t^*,T]$,  $Z(t)>0$. If $Z(\tau)=0$ at any $\tau\in (t^*,T)$, as $u_Z(t)=0$ in this interval and from \eqref{eq:Z_killing} we will have $\dot{Z} =Z (\gamma \beta S {- \pi \beta G u_{h}})<M_7 Z$ for $t \in [\tau, T]$ and for some $M_7>0$ which is an upper-bound to $\gamma \beta S$. This leads to $Z(t)\leq Z(\tau)e^{M_7(t-\tau)}=0$, or $Z(t)=0$ for all $t\in[\tau,T]$ and especially $Z(T)=0$ which is a contradiction. The same reasoning also applies to $t=t^*$ due to continuity. So for $t\in [t^*,T]$,  $Z(t)>0$. Thus, from \eqref{eq:h_cases_1} and the definition of $t^*$, we have $u_h(t)=1$ for all $t\in(t^*,T)$.

So {\bf if $t^*=0$}, we have $u_P(t)=1$, $u_Z(t)=0$, and $u_h(t)=1$ for all $t\in[0,T)$, which agrees with Theorem \ref{thm:killing}. Now {\bf we consider $t^*>0$}.

\subsubsection{Time interval leading up to $t^*>0$ and the existence of $t''$}
{We now look at the evolution of $\varphi_Z$, $\varphi_P$, and $\varphi_h$ for a time interval leading to $t^*>0$, and show $t''$ exists such that for $t\in(t'',t^*)$ we must have $u_P(t)=0$, $u_h(t)=0$, and $u_Z(t)=1$.}
For $t \in (t^*,T)$, and after replacing optimal controls, \eqref{summary_halt} becomes:
\begin{subequations}
\begin{align}
\dot{\varphi}_P &=-\beta G S f'(Z+P)+ \gamma \beta Z (\varphi_Z-\varphi_P) \\
\dot{\varphi}_Z &=\beta G S (k_g-f'(Z+P)) +\beta G (1-\pi) (\varphi_P-\varphi_Z) \notag\\&~~- \gamma \beta S \varphi_Z\label{summary_halt_Z}\\
\dot{\varphi}_P - \dot{\varphi}_Z &=- (\varphi_P-\varphi_Z) ( \gamma \beta Z + \beta G(1-\pi))- \beta G S k_g\notag \\&~~+ \gamma \beta S \varphi_Z\label{summary_halt_ZP}
,\\
\dot{\varphi}_{h}
&= \pi Z (\gamma \beta\varphi_P- \beta G k_g)
.
\end{align}
\end{subequations}
It can be seen that $\dot{\varphi}_P(t)<0$ for $t \in (t^*,T)$ (as ${\varphi}_P (t) > {\varphi}_Z (t)$ and $f'(Z(t)+P(t))>0$ in this interval). This, coupled with $\varphi_P(T)=0$ (\eqref{eq:terminals}) leads to $\varphi_P(t^*)>0$ due to continuity and an integral argument. Thus, we must have $\varphi_Z(t^*)=\varphi_P(t^*)>0$ for $t^*>0$. 

For $t\in(t^*,T)$:
\begin{align}\label{eq:ZandP}
\dot{Z} + \dot{P} =\beta G S + \gamma \beta Z S > 0.
\end{align}
Now, {\bf if $ k_g- f'((Z+P)(t^*))\geq0$}, then $k_g- f'((Z+P)(t))\geq0$ for all $t\in(t^*,T)$ due to the convexity of $k_g - f(\cdot)$ in its argument and as $Z+P$ is strictly increasing in this interval (from \eqref{eq:ZandP}). From \eqref{summary_halt_Z}, $\dot{\varphi}_Z > - \gamma \beta S \varphi_Z \geq -M_8 \varphi_Z$ for all $t\in(t^*,T)$, with $M_8$ being an upper-bound on $\gamma \beta S$. Therefore, $\varphi_Z(t^*)<\varphi_Z(T)e^{-M_8(t^*-T)}=0$ due to an integral argument, which means that ${\varphi}_P (t^*) >0 \geq {\varphi}_Z (t^*)$. This contradicts the starting assumption of this argument, which was  ${\varphi}_P (t^*) = {\varphi}_Z (t^*)$.

Therefore, from here on we will {\bf consider $k_g< f'((Z+P)(t^*))$}. At such a point $t^*$, from \eqref{summary_halt_Z} and the continuity of the states and co-states:
\begin{align}\label{eq:difference2}
(\dot{\varphi}_P (t^{*+})- \dot{\varphi}_Z(t^{*+}))=\beta S(t^*) [\gamma \varphi_Z(t^*) - G  k_g].
\end{align}

Now, \eqref{eq:difference2} should be positive, because if this derivative was strictly negative, the definition of the right-derivative would show that $\varphi_Z(t)>\varphi_P(t)$ for $t$ in an interval starting from $t^*$, a contradiction with the definition of $t^*$. So, as $S(t^*)>0$ from Theorem \ref{thm:constraints}:
\begin{align}\label{inequality_halt}
 \beta S(t^*) [\gamma \varphi_Z(t^*)-Gk_g]\geq0
\hence  \varphi_Z(t^*) \geq \dfrac{ G k_g}{\gamma}.
\end{align}
Now, we can see from a continuity argument on \eqref{summary_halt_Z} (given that $\varphi_Z(t^*)=\varphi_P(t^*)>0$) that $\dot{\varphi}_Z (t^{*-})<0$. Thus $\varphi_Z(t)>\varphi_Z(t^*)>0$ for some interval leading up to $t^*$ due to the definition of a left-derivative. Thus, from \eqref{inequality_halt} we must have:
$\varphi_Z(t) > \dfrac{G k_g}{\gamma}$ (and therefore also $\varphi_Z(t) > 0$) for $t$ in some interval leading up to $t^*$. Let $(t',t^*)$ be the maximal such interval. In this interval, from \eqref{summary_halt_ZP}, we have \[\dot{\varphi}_P - \dot{\varphi}_Z >- (\varphi_P-\varphi_Z) ( \gamma\beta Z + \beta G (1-\pi))\geq - M_9 (\varphi_P-\varphi_Z),\] where $M_9>0$ is an upper-bound on the continuous expression $\gamma\beta Z + \beta G (1-\pi)$. So for any $t$ in this interval, $(\varphi_P(t)-\varphi_Z(t))<(\varphi_P(t^*)-\varphi_Z(t^*))e^{-M_9(t-t^*)}=0$. 

Thus, ${\varphi}_P (t)<{\varphi}_Z (t)$ for $t \in (t',t^*)$. As $\varphi_Z(t^*)>0$, due to the continuity of the states and co-states, there exists a maximal interval $(t'',t^*)$ such that ${\varphi}_Z(t)>\max\{{\varphi}_P (t),0\}$. Following from \eqref{eq:c2} , for $t\in(t'',t^*)$ we must have $u_P(t)=0$ and $u_Z(t)=1$. 

As ${\varphi}_Z(t)>{\varphi}_P (t)$, from \eqref{eq:h_cases_0} and \eqref{eq:h_cases_2} we have $Z(t)u_h(t)=0$ for $t\in(t'',t^*)$.
 This leads to $\dot{Z}(t)>0$ in this interval (from \eqref{eq:Z_killing}), which combined with Theorem \ref{thm:constraints} leads to $Z(t)>0$ in this interval. Therefore, from \eqref{eq:h_cases_0} we can also conclude that in this interval, $u_h(t)=0$.

\subsubsection{Proof of $t''=0$}
If $t'' =0$, this concludes our specification of the structure, which agrees with Theorem \ref{thm:killing}. Thus, henceforth we {\bf consider the case where $t''>0$}, and thus {\bf either $\varphi_Z(t'')={\varphi}_P (t'')$ or ${\varphi}_Z(t'')=0$}.

For $t \in (t{''}, t^*)$, \eqref{summary_halt} becomes:
\begin{subequations}
\begin{align}
\dot{\varphi}_P &=-\beta G S f'(Z+P)+\beta G  (\varphi_Z-\varphi_P)  + \gamma \beta Z (\varphi_Z-\varphi_P)\notag \\
\dot{\varphi}_Z &=\beta G S (k_g-f'(Z+P)) - \gamma \beta S \varphi_Z \label{eq:killing_Z2}\\
\dot{\varphi}_P &- \dot{\varphi}_Z ={- (\varphi_P-\varphi_Z) (\beta G + \gamma \beta Z)- \beta G S k_g+ \gamma \beta S \varphi_Z}\label{eq:killing_ZP2}\\
\dot{\varphi}_{h}
&= - \pi \beta G Z k_g  + \pi \beta G (\varphi_P-\varphi_Z) + \pi \gamma \beta Z \varphi_{P}\notag,
\end{align}
\end{subequations}
Now, for $t\in(t'',t^*)$, 
\begin{align}\label{eq:kg}
k_g-f'((Z+P)(t)) < k_g-f'((Z+P)(t^*))<0
\end{align}  
as $k_g-f(\cdot)$ is convex and in this interval and $\dot{P}(t) + \dot{Z}(t) = \dot{Z}(t) = \beta G S + \gamma \beta Z S> 0$ as $u_Z(t)=1$, and $u_P(t)= u_h(t)=0$. So from \eqref{eq:killing_Z2}, $\dot{\varphi}_Z <- \gamma\beta S \varphi_Z\leq -M_{10} \varphi_Z$ with $M_{10}>0$ being the upper-bound of the continuous $\gamma\beta S$, and therefore for all $t\in (t'',t^*)$,  
$\varphi_Z(t)\geq\varphi_Z(t^*)e^{-M_{10}(t-t^*)}$. As $\varphi_Z(t^*)>0$, $\varphi_Z(t)$ is bounded away from zero,
which {\bf leads to $\varphi_Z(t{''})>0$} due to continuity.

 So we {\bf must have ${\varphi}_P (t{''}) = {\varphi}_Z (t{''})$}. In this case, from \eqref{eq:killing_ZP2} we have
 $(\dot{\varphi}_P  (t''^{+}) - \dot{\varphi}_Z (t''^{+}))\leq 0$, as if it is strictly positive, an integral argument will lead to a contradiction with ${\varphi}_P (t) <{\varphi}_Z (t)$ for $t \in (t{''}, t^*)$. Using the continuity of the states and co-states, as well as the fact that $S(t'')>0$ from Theorem \ref{thm:constraints}, \eqref{eq:simple_Z2} becomes: $
\dot{\varphi}_P  (t''^{+}) - \dot{\varphi}_Z (t''^{+}) = \beta S(t'') [\gamma\varphi_Z(t{''})- Gk_g]\leq 0$ and so: 
\begin{align}
  \varphi_Z(t{''}) \leq \dfrac{G k_g}{\gamma},\label{inequality2_halt}
\end{align}
From \eqref{eq:kg} and \eqref{eq:killing_Z2}, $\dot{\varphi}_Z<-\gamma\beta S\varphi_Z < -M_{10} \varphi_Z<0$. So,
\begin{align}\label{inequality_Z2}
\varphi_Z(t'')>\varphi_Z(t^*).
\end{align}
But \eqref{inequality_halt} and \eqref{inequality2_halt} lead to $\varphi_Z(t'')\leq\varphi_Z(t^*)$, which contradicts \eqref{inequality_Z2}.

Thus $t{''}=0$, and this concludes our specification of the structure of the optimal controls which conform to the structure set out in Theorem \ref{thm:killing}. 
\end{proof}

\vspace{-0.15in}

\section{Proof of Theorem \ref{thm:simple} for adaptive defense model}\label{sec:Optimal_beta1}

We first provide a general framework (akin to the one presented for Theorem \ref{thm:simple}), and then we differentiate the analysis based on the type of adaptive defense used by the network: Constant $\beta(Z)$ in \S Appendix \ref{sec:constant},
affine $\beta(Z)$ in \S Appendix \ref{sec:affine}, and sigmoid $\beta(Z)$ in \S Appendix \ref{sec:sigmoid}

As before, define the continuous co-states $(\lambda_S, \lambda_P, \lambda_Z, \lambda_0)$ such that at points of continuity of the controls: 
  \begin{align}\label{costates_betaZ}
\dot{\lambda}_S &= \beta(Z)[(\lambda_S-\lambda_P)  G u_P+ (\lambda_S-\lambda_Z)  (G u_Z + \gamma Z)]\notag\\
\dot{\lambda}_Z &= -\lambda_0f'(Z+P) +(\lambda_S-\lambda_Z) \gamma \beta(Z) S\notag \\&~~+\beta'(Z) [(\lambda_S-\lambda_P) G S u_P+ (\lambda_S-\lambda_Z) G S u_Z\notag\\&~~ +
(\lambda_S-\lambda_Z) \gamma Z S]\notag\\
\dot{\lambda}_P &=  -\lambda_0f'(Z+P),
\end{align}
with final co-state constraints:
\begin{align}\label{terminal_betaZ}
\lambda_S(T) = \lambda_Z(T) = \lambda_P(T) = 0.
\end{align}
To characterize optimal controls, we define the {\em Hamiltonian}:
\begin{align}\label{Hamiltonian_betaZ}
\ham (t) := \lambda_0f(Z+P)+(\lambda_P-\lambda_S) \beta(Z) G S u_P\notag \\+ (\lambda_Z-\lambda_S) \beta(Z) G S u_Z +
(\lambda_Z-\lambda_S) \gamma\beta(Z) Z S
\end{align}

Pontryagin's Maximum Principle \cite[p.182]{seierstad1987optimal} gives us the following necessary conditions for optimality for an optimal control vector $u^*$:
\begin{flalign}
 &(\lambda_S, \lambda_P, \lambda_Z, \lambda_0)\neq \vec{0}, ~~~ \lambda_0 \in \{0,1\}\label{vec_neq_zero2},
\\& \forall_{u \in \mathcal{U}, t\in[0,T]}~\ham({ S}^*, { Z}^*, { P}^*, u^*, \lambda_S(t), \lambda_P(t), \lambda_Z(t), \lambda_0, t) \geq\notag\\ &\ham({ S}^*, { Z}^*, { P}^*, u, \lambda_S(t), \lambda_P(t), \lambda_Z(t), \lambda_0, t)\label{maximization2}.
\end{flalign}

But if $\lambda_0=0$, $(\lambda_S(T), \lambda_P(T), \lambda_Z(T), \lambda_0)= \vec{0}$, a contradiction, so $\lambda_0=1$.

\vspace{-0.1in}
\subsection{General structure of the optimal control}
If we define:
\begin{subequations}\label{phis_betaZ}
\begin{align}
\varphi_P = (\lambda_P-\lambda_S) \beta(Z) G S \label{phip_betaZ}\\
\varphi_Z = (\lambda_Z-\lambda_S) \beta(Z) G S\label{phiz_betaZ},
\end{align}
\end{subequations}
then, the Hamiltonian becomes:
\begin{align*}
\ham (t) = f(Z+P)+ \varphi_P u_P + \varphi_Z u_Z+
(\lambda_Z-\lambda_S) \gamma\beta(Z) Z S.
\end{align*}

The maximization of the Hamiltonian \eqref{maximization2}, added to the sum constraints for the controls \eqref{sum_bound_nohalt}, leads to \eqref{maxi_nohalt} as the optimality conditions for the controls:\hide{, with the obvious corollary {that $\varphi_P(t) u_P(t)\geq 0$,  and $\varphi_Z(t) u_Z(t)\geq 0$ for all $t$. Further, $[\varphi_Z(t)- \varphi_P(t)] u_Z(t)\geq 0$, and $[\varphi_P(t)- \varphi_Z(t)] u_P(t)\geq 0$ for all $t$.} Also,}
\begin{align}\label{sum_betaZ}
\varphi_Z(t)>0~ \text{or}~\varphi_P(t)>0 \hence u_P(t) + u_Z(t) =1,
\end{align} 
as if that is not true, we can add to the instantaneous value of $\ham(t)$ by adding to either $u_P(t)$ or $u_Z(t)$, a contradiction with the Hamiltonian maximization condition \eqref{maximization2}.

From \eqref{phis_betaZ} and the state \eqref{model:betaZ} and costate \eqref{costates_betaZ} evolution equations and after some manipulation, we have:
\begin{subequations}\label{summary_betaZ}
\begin{align}
\dot{\varphi}_P&={-\beta(Z) G S f'(Z+P)+ \beta'(Z) S \varphi_P[Gu_Z + \gamma Z] }\notag\\&~ {- (\varphi_P- \varphi_Z) \beta(Z) [ G u_Z +\gamma Z] }\label{summary_betaZ_P}\\
\dot{\varphi}_Z &= {-  \beta(Z) G S f'(Z+P) - \varphi_PG u_P\beta'(Z) S }\notag\\&~ { -  \varphi_Z \beta(Z) \gamma S + (\varphi_P-\varphi_Z)  \beta(Z) G u_{P}}\label{summary_betaZ_Z}
,\\
\dot{\varphi}_P& - \dot{\varphi}_Z =  {- (\varphi_P- \varphi_Z) \beta(Z)[G (u_Z + u_P) +\gamma (Z + S)]  }\notag\\&~ {+  \varphi_PS \big[\gamma \beta(Z) + \beta'(Z) [G (u_Z + u_P) + \gamma Z]\big]}\label{summary_betaZ_ZP}
\end{align}
\end{subequations}
{Again, the proof follows the outline laid out in \S \ref{subsec:methodology} (i.e., proving the existence of $t^*$ and $t'$, which are, however, defined differently, and proving $t'=0$ for $t^*>0$), with the difference that the algebraic expressions for $\dot{\varphi}_Z$ and $\dot{\varphi}_P$, and therefore all subsequent analytical arguments, will change.}
\subsubsection{Time interval leading up to $T$ and the existence of $t^*$}
{We follow the evolution of $\varphi_Z$ and $\varphi_P$ for a time interval leading to $T$  and prove the existence of $t^*$ such that we have $u_P(t)=1$, and $u_Z(t)=0$ for all $t\in(t^*,T)$.}

From the terminal time costate conditions \eqref{terminal_betaZ} and their directional derivatives \eqref{summary_betaZ}, we have:
\begin{subequations}
\begin{align}
&\varphi_P(T) = \varphi_Z (T) = 0,\label{eq:terminals2}\\
&\dot{\varphi}_P(T^-)=\dot{\varphi}_Z (T^-)= -\beta(Z) G S f'(Z+P) < 0.
\end{align}
\end{subequations}
So, due to continuity of the states and co-states, there is an interval leading up to T, over which we have $\varphi_P(t)>0$ and $\varphi_Z(t)>0$. Let $(t^*,T)$ be the maximal length interval with this property. For $t\in (t^*, T)$, equation \eqref{sum_betaZ} leads to 
\begin{align}\label{eq:sum1}
u_Z(t)+u_P(t)=1.
\end{align}

Now, for $t \in (t^*,T)$, \eqref{summary_betaZ_ZP} becomes:
\begin{align}\label{eq:PZ_t}
\dot{\varphi}_P(t) - \dot{\varphi}_Z (t) =&  - (\varphi_P- \varphi_Z) \beta(Z)[G +\gamma (Z + S)]\notag\\~&  +  \varphi_PS \big[\gamma \beta(Z) + \beta'(Z) [G  + \gamma Z]\big]
\end{align}
{The rest of the analysis depends on the $\beta(Z)$ function - we present different arguments for $\beta(Z)$'s that are constant, affine, and sigmoid (\S Appendices \ref{sec:constant}, \ref{sec:affine}, and \ref{sec:sigmoid}, respectively). For the affine case (\S \ref{sec:affine}), the analysis needs to be broken down into different cases according to the value of $Z(T)$ in relation to the constant $\frac{1}{2}[\dfrac{\beta_{\max}}{a} - \dfrac{G}{\gamma}]$. When $\beta(Z)$ is a sigmoid (\S \ref{sec:sigmoid}), we use different analytical arguments to prove the result depending on whether $e^{ \alpha(Z(T) - Z_{th})} (1- \frac{\alpha}{\gamma} G- \alpha Z(T)) +1$ is less than, equal to, or greater than zero. For the simple case of constant $\beta(Z)$ (\S \ref{sec:constant}), no such conditional arguments are needed.}

\subsection{Constant $\beta(Z)$}\label{sec:constant}
Assume $\beta(Z)=\beta$.\footnote{Note that this is a case of the model in \S \ref{sec:Optimal_beta} with $g\equiv 0$.} In this case, there is no penalty for creating zombies, and we expect zombies to be created for the whole time period.
Then for $t\in(t^*,T)$, \eqref{eq:PZ_t} becomes:$
(\dot{\varphi}_P - \dot{\varphi}_Z) (t) =  \varphi_PS \gamma \beta-(\varphi_P- \varphi_Z) \beta[G +\gamma (Z + S)]
\geq  - (\varphi_P- \varphi_Z) M_{11}$,
for all $t\in (t^*,T)$ and for some $M_{11}>0$ that is an upper-bound for $\beta(G +\gamma (Z + S))$, as $\varphi_P(t)S(t) \gamma \beta>0$ in this interval. Therefore, for $t\in (t^*, T)$, $\varphi_P(t)- \varphi_Z(t)< [\varphi_P(T)- \varphi_Z(T)]e^{-M_{11}(t-T)}=0$ (from \eqref{eq:terminals2}), and thus $\varphi_P(t) < \varphi_Z(t)$ for $t\in(t^*,T)$. 

Due to the continuity of the states and co-states and from the definition of $t^*$, there exists an interval $(t',T)$, with $t'\leq t^*$ such that  $\varphi_Z>\varphi_P$ and $\varphi_Z>0$. These conditions, coupled with \eqref{eq:c} lead to $u_P(t)=0$ and $u_Z(t)=1$ for all $t \in (t', T)$.

We now prove $t'=0$. If this does not hold, either $\varphi_Z(t') = \varphi_P(t')$ or $\varphi_Z(t') = 0$ for some $t'>0$ due to continuity of the states and co-states.

Since $u_P(t)=0$ for $t\in (t',T)$, \eqref{summary_betaZ_Z} becomes: $\dot{\varphi}_Z(t) = -  \beta(Z) G S f'(Z+P)  -  \varphi_Z \beta(Z) \gamma S <0$,
which leads to $\varphi_Z(t')>\varphi_Z(T) = 0$. Thus, $\varphi_Z(t')$ cannot be equal to zero.

{\bf If} $\varphi_Z(t')=\varphi_P(t')$, then from \eqref{summary_betaZ_ZP}, $\beta'(Z)=0$ for constant $\beta(Z)$, and the continuity of the states and co-states: $
\big(\dot{\varphi}_P - \dot{\varphi}_Z\big) (t'^+) =  \varphi_P(t')S(t') \gamma \beta
= \varphi_Z(t')S(t') \gamma \beta>0$,
leading to the existence of an interval $(t', t'')$ over which $\varphi_P(t)>\varphi_Z(t)$, a contradiction with the definition of $t'$. 

Thus, $t'=0$ and $u_Z(t)=1$ and $u_P(t)=0$ for all $t$, which agrees with the statement of Theorem \ref{thm:simple} and our intuition that zombies will be created for the entire period.  \qed

\subsection{Affine $\beta(Z)$}\label{sec:affine}

Assume $\beta(Z) = - a Z + \beta_{\max}$, with $0<a\leq \beta_{\max}$ (as $\beta_{\max}$ is an upperbound on this $\beta(Z)$ and $\beta(Z)>0$). Then, for $t \in (t^*,T)$, \eqref{eq:PZ_t} becomes:
\begin{align}\label{eq:difference3}
\dot{\varphi}_P(t)& - \dot{\varphi}_Z(t) = -  a\varphi_PS \big[\gamma(2Z-\dfrac{\beta_{\max}}{a}) + G\big]\notag\\&~~ {- (\varphi_P- \varphi_Z) (-a Z + \beta_{\max})[G +\gamma (Z + S)]}
\end{align}

Now we break down the situations that can arise based on the value of $Z(T)$ with respect to the fixed $\frac{1}{2}[\frac{\beta_{\max}}{a} - \frac{G}{\gamma}]$:

\subsubsection{$Z(T)\leq\frac{1}{2}[\dfrac{\beta_{\max}}{a} - \dfrac{G}{\gamma}]$}\label{part:I} Note that for this case, we must have $\frac{1}{2}[\frac{\beta_{\max}}{a} - \frac{G}{\gamma}]\geq0$ due to Theorem \ref{thm:constraints}.

We first consider the sub-case where $Z(T)=\frac{1}{2}[\frac{\beta_{\max}}{a} - \frac{G}{\gamma}]=0$. Here, we must have $\dot{Z}(t)=0$ for all $t$ as $\dot{Z}(t)\geq 0$ for all $t$ and as states are continuous. The only way for $\dot{Z}(t)=0$ for all $t$ is for us to have $Z_0=0$ and $u_Z(t)=0$ for all $t<T$ (due to Theorem \ref{thm:constraints}). This leads to \eqref{summary_betaZ_P} becoming $\dot{\varphi}_P(t) ={-\beta(0) G S(t) f'(P(t))<0}$ for all $t<T$, and thus $\varphi_P(t)>0$. This fact, combined with $u_Z(t)=0$ for all $t$ and \eqref{eq:b} leads to $u_P(t)=1$ for all $t$ (i.e., $t^*=0$ in the statement of Theorem \ref{thm:simple}).

Otherwise, we either have (i) $Z(T)=\frac{1}{2}[\frac{\beta_{\max}}{a} - \frac{G}{\gamma}]>0$ or (ii) $Z(T)<\frac{1}{2}[\frac{\beta_{\max}}{a} - \frac{G}{\gamma}]$. 

(i) In this case, from \eqref{model:betaZ} (for which $\beta(Z)>0$ and $G>0$), Theorem \ref{thm:constraints} (which specifies $S(T)>0$), and continuity of the states, we have $\dot{Z}(T^-)>0$. Thus $Z(t)<\frac{1}{2}[\frac{\beta_{\max}}{a} - \frac{G}{\gamma}]$ for some $(t'',T)$. Therefore, as $\dot{Z}(t)\geq0$ for all $t$, so $Z(t)<\frac{1}{2}[\frac{\beta_{\max}}{a} - \frac{G}{\gamma}]$ for all $t<T$.

(ii) Since $\dot{Z}\geq 0$ from \eqref{model:betaZ} and Theorem \ref{thm:constraints}, in this case we also have $Z(t)<\frac{1}{2}[\frac{\beta_{\max}}{a} - \frac{G}{\gamma}]$ for all $t<T$.

Therefore for both (i) and (ii), $\gamma\beta_{\max}-2\gamma aZ(t) - Ga>0$ for all $t<T$. 

 From \eqref{eq:difference3} and for all $t\in (t^*,T)$: $\dot{\varphi}_P(t) - \dot{\varphi}_Z (t) >  - (\varphi_P- \varphi_Z) \beta(Z)[G +\gamma (Z + S)]\geq  - (\varphi_P- \varphi_Z) M_{12}$,
for some $M_{12}>0$ which is an upper-bound to the continuous $\beta(Z)[G +\gamma (Z + S)]$ over this interval. Therefore, for $t\in (t^*, T)$, $\varphi_P(t)- \varphi_Z(t)< [\varphi_P(T)- \varphi_Z(T)]e^{-M_{12}(t-T)}=0$, and thus $\varphi_P(t) < \varphi_Z(t)$ for $t\in(t^*,T)$. 

Due to the continuity of the states and co-states and because for $t\in(t^*,T)$, $\varphi_Z(t)>0$, there exists an interval $(t',T)$, with $t'\leq t^*$ such that both $\varphi_Z(t)>\varphi_P(t)$ and $\varphi_Z(t)>0$. These conditions, coupled with \eqref{eq:c} lead to $u_P(t)=0$ and $u_Z(t)=1$ for all $t \in (t', T)$.

We now prove $t'=0$. If this does not hold, either $\varphi_Z(t') = 0$  or $\varphi_Z(t') = \varphi_P(t')$ for some $t'>0$ due to continuity of the states and co-states.

For $t\in (t',T)$ \eqref{summary_betaZ_Z} becomes, $
\dot{\varphi}_Z(t) =
-  \beta(Z) G S f'(Z+P)  -  \varphi_Z \beta(Z) \gamma S <0$,
which leads to $\varphi_Z(t')>\varphi_Z(T) = 0$.

So we must have $\varphi_Z(t')=\varphi_P(t')$ for $t'>0$. From \eqref{eq:difference3} and the continuity of the states and co-states, $\big(\dot{\varphi}_P - \dot{\varphi}_Z\big) (t'^+) =
\varphi_P(t')S(t') \big[\gamma\beta_{\max}-2\gamma aZ(t') - Ga\big]
=\varphi_Z(t')S(t') \big[\gamma\beta_{\max}-2\gamma aZ(t') - Ga\big]
>0$,
leading to the existence of an interval $(t', t'')$ over which $\varphi_P(t)>\varphi_Z(t)$, a contradiction with the definition of $t'$. 

Thus, $t'=0$ and $u_Z(t)=1$ and $u_P(t)=0$ for all $t$, which agrees with the statement of Theorem \ref{thm:simple}.
 \qed

\subsubsection{$Z(T)>\frac{1}{2}[\dfrac{\beta_{\max}}{a} - \dfrac{G}{\gamma}]$} Due to the continuity of the states, $Z(t)>\frac{1}{2}[\frac{\beta_{\max}}{a} - \frac{G}{\gamma}]$ for $t\in (t_1,T)$ for some $t_1$. Recall that for $t\in(t^*,T)$, $\varphi_P(t)>0$. Thus, for $t\in (t_2,T)$, where $t_2=\max\{t^*,t_1\}$ and with $M_{12}$ again defined as the upper-bound to the continuous $\beta(Z)[G +\gamma (Z + S)]$, \eqref{eq:difference3} leads to:
$\dot{\varphi}_P(t) - \dot{\varphi}_Z (t) <  - (\varphi_P- \varphi_Z) \beta(Z)[G  +\gamma (Z + S)]\leq  - (\varphi_P- \varphi_Z) M_{12}$.
Therefore, in this interval, $\varphi_P(t)- \varphi_Z(t)> [\varphi_P(T)- \varphi_Z(T)]e^{-M_{12}(t-T)}=0$, and thus $\varphi_P(t) > \varphi_Z(t)$ and $\varphi_P(t)>0$ for $t\in(t_2,T)$. 

Now, due to the continuity of the states and co-states, define $(t_3,T)$ to be the maximal length interval over which $\varphi_P(t) > \max\{ \varphi_Z(t), 0\}$. Note that for $t\in(t_3,T)$ we have $u_Z(t)=0$ and $u_P(t)=1$ due to \eqref{eq:b}.

Due to continuity of the states and co-states, either $t_3=0$, in which case $u_Z(t)=0$ and $u_P(t)=1$ for all $t$ (agreeing with the structure of Theorem \ref{thm:simple}), or we have a $t_3>0$ such that $\varphi_P(t_3)=0$ or $\varphi_P(t_3) = \varphi_Z(t_3)>0$.
 
From \eqref{summary_betaZ_P}, Theorem \ref{thm:constraints}, and from the definition of $t_3$, for $t\in (t_3,T)$ we have, $
\dot{\varphi}_P=-\beta(Z) G S f'(Z+P)-(\varphi_P- \varphi_Z) \beta(Z) \gamma Z  - a S \varphi_P \gamma Z <  - a S \varphi_P \gamma Z \leq - M_{13} \varphi_P$,
for some $M_{13}>0$ that is an upper-bound to the continuous $a _1 S \gamma Z$ over this interval. Thus, $\varphi_P(t_3)>\varphi_P(T)e^{-M_{13}(t_3 - T)}=0$.
So for $t_3>0$, we must have $\varphi_P(t_3) = \varphi_Z(t_3)>0$. From the continuity of the states and co-states, there must exist an interval leading up to $t_3$ such that $\varphi_Z(t)>0$ and $\varphi_P(t)>0$. Let $(t_4,t_3)$ be the maximal-length interval with such a property. Notice that \eqref{sum_betaZ} also applies, leading to $u_P(t)+u_Z(t)=1$ for $t\in (t_4, t_3)$.

Furthermore, also from continuity, \eqref{eq:difference3} becomes:\small
\begin{align}\label{eq:difference4}
\hspace*{-0.1in}(\dot{\varphi}_P& - \dot{\varphi}_Z)(t_3^+) = -  a\varphi_P(t_3) S(t_3) \big[\gamma(2Z(t_3)-\dfrac{\beta_{\max}}{a}) + G\big]
\end{align}\normalsize
But if $\dot{\varphi}_P(t_3^+) - \dot{\varphi}_Z(t_3^+)<0$, then due to continuity and the definition of the derivative, we must have an interval starting from $t_3$ where $\varphi_Z(t)>\varphi_P(t)$, which contradicts the definition of $t_3$ (which stated that over an interval starting at $t_3$, $\varphi_P(t)>\max\{\varphi_P(t),0\}$). So we must have $\dot{\varphi}_P(t_3^+) - \dot{\varphi}_Z(t_3^+) \geq 0$. From \eqref{eq:difference4} this is equivalent to 
$[\gamma(2Z(t_3)-\frac{\beta_{\max}}{a}) + G\big]\leq 0$, or $Z(t_3)\leq \frac{1}{2}[\frac{\beta_{\max}}{a} - \frac{G}{\gamma}]$. 

Following the same set of arguments as presented in \S \ref{part:I}  for the case of $Z(T)\leq \frac{1}{2}[\frac{\beta_{\max}}{a} - \frac{G}{\gamma}]$ and retracing them for $Z(t_3)\leq \frac{1}{2}[\frac{\beta_{\max}}{a} - \frac{G}{\gamma}]$ (with $t_3$ replacing $T$ in all arguments) shows that the structure postulated in Theorem \ref{thm:simple} holds.

Thus, all possible state and co-state trajectories lead to the structure postulated in Theorem \ref{thm:simple}.\qed

\vspace{-0.1in}
\subsection{Sigmoid $\beta(Z)$}\label{sec:sigmoid}

Assume $\beta_Z =  \dfrac{\beta_0}{1+e^{ \alpha(Z - Z_{th})}}$, with $0<Z_{th}<1$ being a fixed threshold and $\alpha>0$ denoting the sharpness of the cut-off. This simulates a threshold-like detection of zombies by a network administrator. In this case, \eqref{summary_betaZ_ZP} becomes:
\begin{align}\label{eq:difference5}
&\dot{\varphi}_P- \dot{\varphi}_Z = {- (\varphi_P- \varphi_Z) \beta(Z)[G (u_Z + u_P) +\gamma (Z + S)]  }\notag\\&\hspace{-0.05in}~ {+  \dfrac{\beta_0\gamma\varphi_PS \big[ e^{ \alpha(Z - Z_{th})} (1- \frac{\alpha}{\gamma} G (u_z + u_P)- \alpha Z) +1\big]}{(1+e^{ \alpha(Z - Z_{th})})^2}}
\end{align}
Define: $\Psi(Z, u_Z + u_P):=e^{ \alpha(Z - Z_{th})} (1- \frac{\alpha}{\gamma} G (u_z + u_P)- \alpha Z) +1$. Then \eqref{eq:difference5} becomes: 
\begin{align} \label{eq:difference6}
\dot{\varphi}_P- \dot{\varphi}_Z &= {- (\varphi_P- \varphi_Z) \beta(Z)[G (u_Z + u_P) +\gamma (Z + S)]  }\notag\\&+  \dfrac{\beta_0\gamma\varphi_PS }{(1+e^{ \alpha(Z - Z_{th})})^2}\Psi(Z, u_Z + u_P)
\end{align}
Now, for possible intervals where $u_Z+u_P$ is a constant $c\in [0,1]$, $\Psi(Z,c)$ is a function of one variable ($Z$).
We can see that at points of continuity of the controls and in intervals where it is defined, $\Psi(Z,c)$ is also continuous and differentiable. Furthermore, we can see that at points of continuity of the controls in these intervals, we have:
\begin{align}\label{eq:psi_dot}
\dfrac{d\Psi(Z,c)}{dZ}&= - \alpha^2 e^{ \alpha(Z - Z_{th})} (\frac{G}{\gamma} c +  Z) <0 
\end{align}
Now we break down the situations that can arise based on the value of $\Psi(Z(T),1)$:

\subsubsection{$\Psi(Z(T),1)>0$}\label{part:III} From $\dot{Z}\geq 0$ (\eqref{model:betaZ} and Theorem \ref{thm:constraints}) and the continuity of the states, we have $Z(t)\leq Z(T)$ for all $t$. Now for $t\in(t^*,T)$, as the sum of the controls is constant and equal to one due to \eqref{eq:sum1},  we will have $\Psi(Z(t),1)\geq \Psi(Z(T),1)>0$ due to \eqref{eq:psi_dot}. Thus from \eqref{eq:difference6} and for all $t\in (t^*,T)$ at which the controls are continuous: $
\dot{\varphi}_P(t) - \dot{\varphi}_Z (t) >  - (\varphi_P- \varphi_Z) \beta(Z)[G +\gamma (Z + S)]\geq  - (\varphi_P- \varphi_Z) M_{14}$,
for some $M_{14}>0$ which is an upper-bound to the continuous $\beta(Z)[G +\gamma (Z + S)]$. Therefore, for $t\in (t^*, T)$, $\varphi_P(t)- \varphi_Z(t)< [\varphi_P(T)- \varphi_Z(T)]e^{-M_{14}(t-T)}=0$, and thus $\varphi_P(t) < \varphi_Z(t)$ for $t\in(t^*,T)$. 

Due to the continuity of the states and co-states and from the definition of $t^*$, there exists an interval $(t',T)$, with $t'\leq t^*$ such that  $\varphi_Z(t)>\varphi_P(t)$ and $\varphi_Z(t)>0$. These conditions, coupled with \eqref{eq:c} lead to $u_P(t)=0$ and $u_Z(t)=1$ for all $t \in (t', T)$, with the corollary that $u_P(t)+ u_Z(t)=1$.

We now prove $t'=0$. If this does not hold, either $\varphi_Z(t') = 0$ or $\varphi_Z(t') = \varphi_P(t')>0$ for $t'>0$ due to continuity of the states and co-states.

For $t\in (t',T)$, as $u_P(t)=0$, \eqref{summary_betaZ_Z} becomes: $\dot{\varphi}_Z(t) =
-  \beta(Z) G S f'(Z+P)  -  \varphi_Z \beta(Z) \gamma S <0$,
as each term in the right hand side is strictly positive in the interval. Now, if we have $\varphi_Z(t') = 0$, from this time-derivative and continuity of the states and co-states we must have $\varphi_Z(t')>\varphi_Z(T) = 0$. Thus, $\varphi_Z(t') = 0$ is ruled out.

On the other hand, if $\varphi_Z(t')=\varphi_P(t')>0$, then from \eqref{eq:difference6} and the continuity of the states and co-states: $
\big(\dot{\varphi}_P - \dot{\varphi}_Z\big) (t'^+) =  \dfrac{\beta_0\gamma\varphi_P(t') S(t') }{(1+e^{ \alpha(Z (t') - Z_{th})})^2}\Psi(Z(t'),1)>0$,
leading to the existence of an interval $(t', t'')$ over which $\varphi_P(t)>\varphi_Z(t)$, a contradiction with the definition of $t'$. 

Thus, $t'=0$ and $u_Z(t)=1$ and $u_P(t)=0$ for all $t$, which agrees with the statement of Theorem \ref{thm:simple}.

\subsubsection{$\Psi(Z(T),1)=0$ and $Z(T)>0$}\label{part:IV} We have $\dot{Z}(T^-)>0$ (from \eqref{model:betaZ}, Theorem \ref{thm:constraints}, and continuity) which leads to $Z(t)< Z(T)$ for an interval leading up to $t$. As $\dot{Z}\geq 0$, we can extend $Z(t)<Z(T)$ to all $t$. Now for $t\in(t^*,T)$, from \eqref{eq:sum1},  we will have $\Psi(Z(t),1) > \Psi(Z(T),1)=0$ due to \eqref{eq:psi_dot}. We now prove  $t'=0$ and $u_Z(t)=1$ and $u_P(t)=0$ for all $t$.

 From  \eqref{eq:sum1}, \eqref{eq:difference6}, for all $t\in (t^*,T)$ (over which $\varphi_P(t)>0$): $\dot{\varphi}_P(t) - \dot{\varphi}_Z (t) >  - (\varphi_P- \varphi_Z) \beta(Z)[G +\gamma (Z + S)]\geq  - (\varphi_P- \varphi_Z) M_{12}$
for some $M_{12}>0$ which is an upper-bound to the continuous $\beta(Z)[G +\gamma (Z + S)]$ over this interval. Therefore, for $t\in (t^*, T)$, $\varphi_P(t)- \varphi_Z(t)< [\varphi_P(T)- \varphi_Z(T)]e^{-M_{12}(t-T)}=0$, and thus $\varphi_P(t) < \varphi_Z(t)$ for $t\in(t^*,T)$. 

Due to the continuity of the states and co-states and because for $t\in(t^*,T)$, $\varphi_Z(t)>0$, there exists an interval $(t',T)$, with $t'\leq t^*$ such that both $\varphi_Z(t)>\varphi_P(t)$ and $\varphi_Z(t)>0$. These conditions, coupled with \eqref{eq:c} lead to $u_P(t)=0$ and $u_Z(t)=1$ for all $t \in (t', T)$.

We now prove $t'=0$. If this does not hold, either (i) $\varphi_Z(t') = 0$  or (ii) $\varphi_Z(t') = \varphi_P(t')$ for some $t'>0$ due to continuity of the states and co-states.

For $t\in (t',T)$ \eqref{summary_betaZ_Z} becomes: $
\dot{\varphi}_Z(t) =
-  \beta(Z) G S f'(Z+P)  -  \varphi_Z \beta(Z) \gamma S <0$,
which leads to $\varphi_Z(t')>\varphi_Z(T) = 0$.

So for $t'>0$ we must have $\varphi_Z(t')=\varphi_P(t')$. From \eqref{eq:difference6} and the continuity of the states and co-states: $
\big(\dot{\varphi}_P - \dot{\varphi}_Z\big) (t'^+) =
\dfrac{\beta_0\gamma\varphi_P(t')S(t') }{(1+e^{ \alpha(Z(t') - Z_{th})})^2}\Psi(Z(t'), 1)=\dfrac{\beta_0\gamma\varphi_Z(t')S(t') }{(1+e^{ \alpha(Z(t') - Z_{th})})^2}\Psi(Z(t'), 1)
>0$,
leading to the existence of an interval $(t', t'')$ over which $\varphi_P(t)>\varphi_Z(t)$, a contradiction with the definition of $t'$. 

Thus, $t'=0$ and $u_Z(t)=1$ and $u_P(t)=0$ for all $t$, which agrees with the statement of Theorem \ref{thm:simple}.
\qed

\subsubsection{$\Psi(Z(T),1)=0$ and $Z(T)=0$}\label{part:V} We must have $\dot{Z}(t)=0$ for all $t$ as $\dot{Z}\geq 0$ and as states are continuous. The only way for $\dot{Z}(t)=0$ for all $t$ is for us to have $Z_0=0$ and $u_Z(t)=0$ for all $t<T$ (due to Theorem \ref{thm:constraints}). This leads to \eqref{summary_betaZ_P} becoming $\dot{\varphi}_P(t) ={-\beta(0) G S(t) f'(P(t))<0}$ for all $t<T$, and thus $\varphi_P(t)>0$. This fact, combined with $u_Z(t)=0$ for all $t$ and \eqref{eq:b} leads to $u_P(t)=1$ for all $t$.

\subsubsection{$\Psi(Z(T),1)<0$} Due to the continuity of the states, $\Psi(Z(t),1)<0$ for $t\in (t_1,T)$ for some $t_1$. Thus, \eqref{eq:difference6} leads to $
\dot{\varphi}_P(t) - \dot{\varphi}_Z (t) <  - (\varphi_P- \varphi_Z) \beta(Z)[G  +\gamma (Z + S)]\leq  - (\varphi_P- \varphi_Z) M_{12}$, 
for $t\in (t_2,T)$, where $t_2=\max\{t^*,t_1\}$ and with $M_{12}$ defined as before (an upper-bound to the continuous $\beta(Z)[G +\gamma (Z + S)]$ over this interval). Therefore, in this interval, $\varphi_P(t)- \varphi_Z(t)> [\varphi_P(T)- \varphi_Z(T)]e^{-M_{12}(t-T)}=0$, and thus $\varphi_P(t) > \varphi_Z(t)$ and $\varphi_P(t)>0$ for $t\in(t_2,T)$. 

Now, due to the continuity of the states and co-states, define $(t_3,T)$ to be the maximal length interval over which $\varphi_P(t) > \varphi_Z(t)$ and $\varphi_P(t)>0$. Note that for $t\in(t_3,T)$ we have (due to \eqref{eq:b}) $u_Z(t)=0$ and $u_P(t)=1$\hide{, and thus $u_P(t)+ u_Z(t)=1$}.

Due to continuity of the states and co-states, either $t_3=0$, in which case $u_Z(t)=0$ and $u_P(t)=1$ for all $t$, or we have a $t_3>0$ such that (i) $\varphi_P(t_3)=0$ or (ii) $\varphi_P(t_3) = \varphi_Z(t_3)>0$. 

From \eqref{summary_betaZ_P}, Theorem \ref{thm:constraints}, and from the definition of $t_3$, for $t\in (t_3,T)$ we have:  $
\dot{\varphi}_P=-\beta(Z) G S f'(Z+P)-(\varphi_P- \varphi_Z) \beta(Z) \gamma Z  - \dfrac{\alpha\beta_0\gamma e^{ \alpha(Z - Z_{th})}}{(1+e^{ \alpha(Z - Z_{th})})^2} S \varphi_P  Z<  - \dfrac{\alpha\beta_0\gamma e^{ \alpha(Z - Z_{th})}SZ}{(1+e^{ \alpha(Z - Z_{th})})^2} \varphi_P \leq - M_{15} \varphi_P$,
for some $M_{15}>0$ that is an upper-bound to the continuous $\dfrac{\alpha\beta_0\gamma e^{ \alpha(Z - Z_{th})}SZ}{(1+e^{ \alpha(Z - Z_{th})})^2}$. Thus, $\varphi_P(t_3)>\varphi_P(T)e^{-M_{15}(t_3 - T)}=0$. 

So for $t_3>0$ we must have $\varphi_P(t_3) = \varphi_Z(t_3)>0$. From the continuity of the states and co-states, there must exist an interval leading up to $t_3$ such that $\varphi_Z(t)>0$ and $\varphi_P(t)>0$. Let $(t_4,t_3)$ be the maximal-length interval with such a property. Notice that \eqref{sum_betaZ} also applies, leading to $u_P(t)+u_Z(t)=1$ for $t\in (t_4, t_3)$.

Furthermore, also from continuity, \eqref{eq:difference6} becomes:
\begin{align}\label{eq:difference7}
\dot{\varphi}_P(t_3^+)& - \dot{\varphi}_Z(t_3^+) = \dfrac{\beta_0\gamma\varphi_P(t_3)S(t_3)}{(1+e^{ \alpha(Z(t_3) - Z_{th})})^2}\Psi(Z(t_3),1)
\end{align}
But if $\dot{\varphi}_P(t_3^+) - \dot{\varphi}_Z(t_3^+)<0$, then due to continuity and the definition of the derivative, we must have an interval starting from $t_3$ where $\varphi_Z(t)>\varphi_P(t)$, which contradicts the definition of $t_3$. So we must have $\dot{\varphi}_P(t_3^+) - \dot{\varphi}_Z(t_3^+) \geq 0$. From \eqref{eq:difference4} this is equivalent to 
$\Psi(Z(t_3),1)\geq 0$. Following the same arguments presented in \S \ref{part:III}, \S \ref{part:IV}, and \S \ref{part:V}  for the case of $\Psi(Z(T),1)\geq0$ and retracing them for $\Psi(Z(t_3),1)\geq0$ (with $t_3$ replacing $T$) shows Theorem \ref{thm:simple}'s structure holds.

Thus, Theorem \ref{thm:simple} holds for all possible trajectories.

\end{appendices}
\end{document}